\documentclass[conference]{IEEEtran}
\IEEEoverridecommandlockouts
\usepackage{cite}
\usepackage{array}
\usepackage{graphicx}
\usepackage{amsmath,amssymb,amsfonts}
\usepackage[center]{caption}
\usepackage[font=footnotesize]{caption}
\usepackage{xcolor}
\usepackage{algorithm,algorithmicx,algpseudocode}
\usepackage{bm}
\usepackage{mathtools}
\usepackage{nccmath}
\definecolor{Green}{rgb}{0.13, 0.55, 0.13}
\usepackage{stfloats}
\usepackage{comment}
\excludecomment{draftnotes}
\usepackage{paralist}
\usepackage{subcaption}
\usepackage{enumitem}
\usepackage{mathrsfs}
\usepackage{balance}

\usepackage{amsthm}
\newtheorem{theorem}{Theorem}
\newtheorem{proposition}{Proposition}
\newtheorem{lemma}{Lemma}

\DeclareMathOperator*{\argmin}{arg\,min}

\title{From Timestamps to Versions: Version AoI in Single- and Multi-Hop Networks}

\author{
		\IEEEauthorblockN{Erfan Delfani and Nikolaos Pappas\\} 
		\IEEEauthorblockA{Department of Computer and Information Science, Link\"{o}ping University, Link\"{o}ping, Sweden}
        Email: \{erfan.delfani, nikolaos.pappas\}@liu.se
        \vspace{-12pt}
        \thanks{This work has been supported in part by the Swedish Research Council (VR), ELLIIT, and the European Union (ELIXIRION, 101120135, 6G-LEADER, 101192080, and SOVEREIGN, 101131481).}
	}

\setlist[enumerate]{nosep}
\setlist[itemize]{nosep}

\begin{document}

    \vspace{-12pt}
	\maketitle
	
	\begin{abstract}
		Timely and informative data dissemination in communication networks is essential for enhancing system performance and energy efficiency, as it reduces the transmission of outdated or redundant data. Timeliness metrics, such as Age of Information (AoI), effectively quantify data freshness; however, these metrics fail to account for the intrinsic informativeness of the content itself. To address this limitation, content-based metrics have been proposed that combine both timeliness and informativeness. Nevertheless, existing studies have predominantly focused on evaluating average metric values, leaving the complete distribution—particularly in multi-hop network scenarios—largely unexplored. In this paper, we provide a comprehensive analysis of the stationary distribution of the Version Age of Information (VAoI), a content-based metric, under various scheduling policies, including randomized stationary, uniform, and threshold-based policies, with transmission constraints in single-hop and multi-hop networks. We derive closed-form expressions for the stationary distribution and average VAoI under these scheduling approaches. Furthermore, for threshold-based scheduling, we analytically determine the optimal threshold value that minimizes VAoI and derive the corresponding optimal VAoI in closed form. Numerical evaluations verify our analytical findings and provide valuable insights into leveraging VAoI in the design of efficient communication networks.
	\end{abstract}

	\section{Introduction}
	Efficient data management is a critical requirement for ensuring optimal performance in communication networks across a wide range of applications, from single-hop IoT monitoring systems to multi-hop satellite-based networks. As the volume of data generated by these networks increases, transmitting all data indiscriminately, without considering its semantic significance or task-specific utility, becomes increasingly unsustainable. Such an approach results in excessive consumption of critical resources, including energy and bandwidth, ultimately compromising system practicality and degrading overall performance. To address these challenges, there is an urgent need for network management approaches that optimize data transmission by leveraging goal-oriented semantic metrics; that is, by delivering the most timely and informative data within a constrained frequency of data transmissions~\cite{kountouris2021semantics, luo2025survey}.
	
	The AoI~\cite{kaul2012real} is a widely used semantic metric that quantifies the freshness of information in status update systems as the time elapsed since the generation of the most recently received data. AoI-aware scheduling effectively minimizes staleness by adapting transmissions to source data arrivals and network service times~\cite{yates2021age}. However, AoI captures freshness solely through data \emph{timestamps}, without accounting for actual changes in the source content. As a result, simply \emph{refreshing timestamps} may fail to deliver new information, and avoiding such redundant updates can reduce data transmission and energy consumption.
	
	To address this limitation, \emph{content-based} metrics such as Age of Incorrect Information (AoII) \cite{maatouk2020age} and VAoI \cite{yates2021vage} have been introduced. AoII adds a distortion-aware dimension by measuring the staleness of incorrect information—specifically when the receiver’s content deviates from the source—unlike AoI, which treats both correct and incorrect data uniformly. However, AoII requires precise knowledge of the content or \emph{state} of information at both source and destination for comparison, which is practical only when the state space is small and fully modeled, with all transitions known. In many real-world applications, such complete knowledge of source content and transitions may not be available.
	In such cases, VAoI offers a more practical, content-based metric by focusing solely on \emph{content changes} at the source, where data evolve through successive, non-reverting versions. This requires minimal knowledge: at any time, either a new or the previous version exists, and the receiver must track these versions as timely as possible. Defined as the number of versions by which the receiver lags behind the source, VAoI further improves upon AoI by replacing timestamps with version numbers, thereby eliminating the challenging requirement of clock synchronization between the transmitter and receiver \cite{Mehrdad2025CL}. It is computed simply by comparing the receiver’s stored version with the source’s current version.
	
    While these metrics have attracted attention, the majority of existing research has focused on first-moment analyses, i.e., average values. However, a deeper understanding of their full distributions is critical for analyzing and optimizing system behavior, particularly under resource constraints.
    This gap becomes even more significant when moving from single-hop to multi-hop communication scenarios.
    In this work, we address this gap by providing a comprehensive analysis of the distribution of the content-based metric VAoI under various rate-constrained transmission policies in both single-hop and multi-hop settings. The main contributions of this study are as follows:
	\begin{inparaenum}[1)]
		\item We derive closed-form expressions for the stationary distribution and average VAoI in a single-hop setup under three source transmission policies, subject to an average update rate constraint.
		\item We investigate an optimal on-off transmission policy under the rate constraint, proving its threshold-based structure and deriving the closed-form optimal threshold and the resulting average VAoI.
		\item We extend the analysis to a multi-hop setup with $N$ intermediate nodes over unreliable links, demonstrating that the destination VAoI equals a time-shifted copy of the first node’s VAoI plus additional random variables, and we derive the corresponding closed-form average VAoI.
		\item We validate the analytical findings through simulations and investigate the behavior of the VAoI in both single-hop and multi-hop scenarios.
	\end{inparaenum}
    
    The remainder of the paper is organized as follows. Section \ref{sec_RelatedWorks} reviews related work, while Section \ref{sec_SystemModel} presents the system model. Section \ref{Sec_OneHop} analyzes VAoI in a single-hop setup under various update policies, followed by a multi-hop VAoI analysis in Section \ref{Sec_MultiHop}. Numerical results are discussed in Section \ref{Sec_NumericResults}, and Section \ref{sec_Conclusion} concludes the paper.

    \begin{figure*}[!t]
		\centering \
		\begin{minipage}[c]{0.355\linewidth}
			\centering
		\includegraphics[scale=0.057]{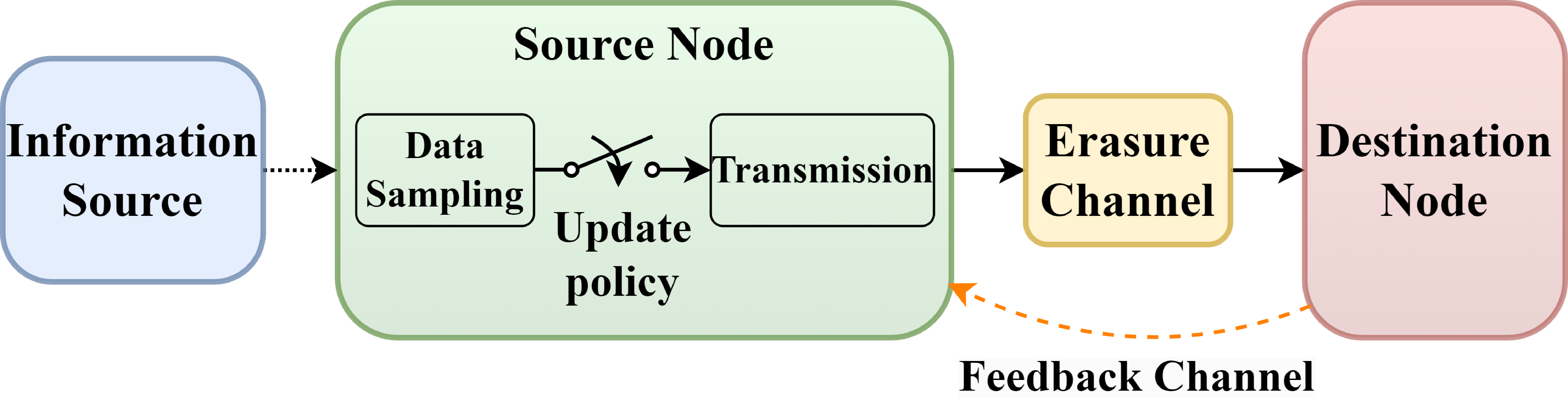} 
		\caption{Single-hop status update system.}
		\label{fig_SingleHopSetup}
		\end{minipage}
		\ \ \hfill
		\begin{minipage}[c]{0.31\linewidth}
			\centering
		\includegraphics[scale=0.088, trim=0 0 140 0, clip]{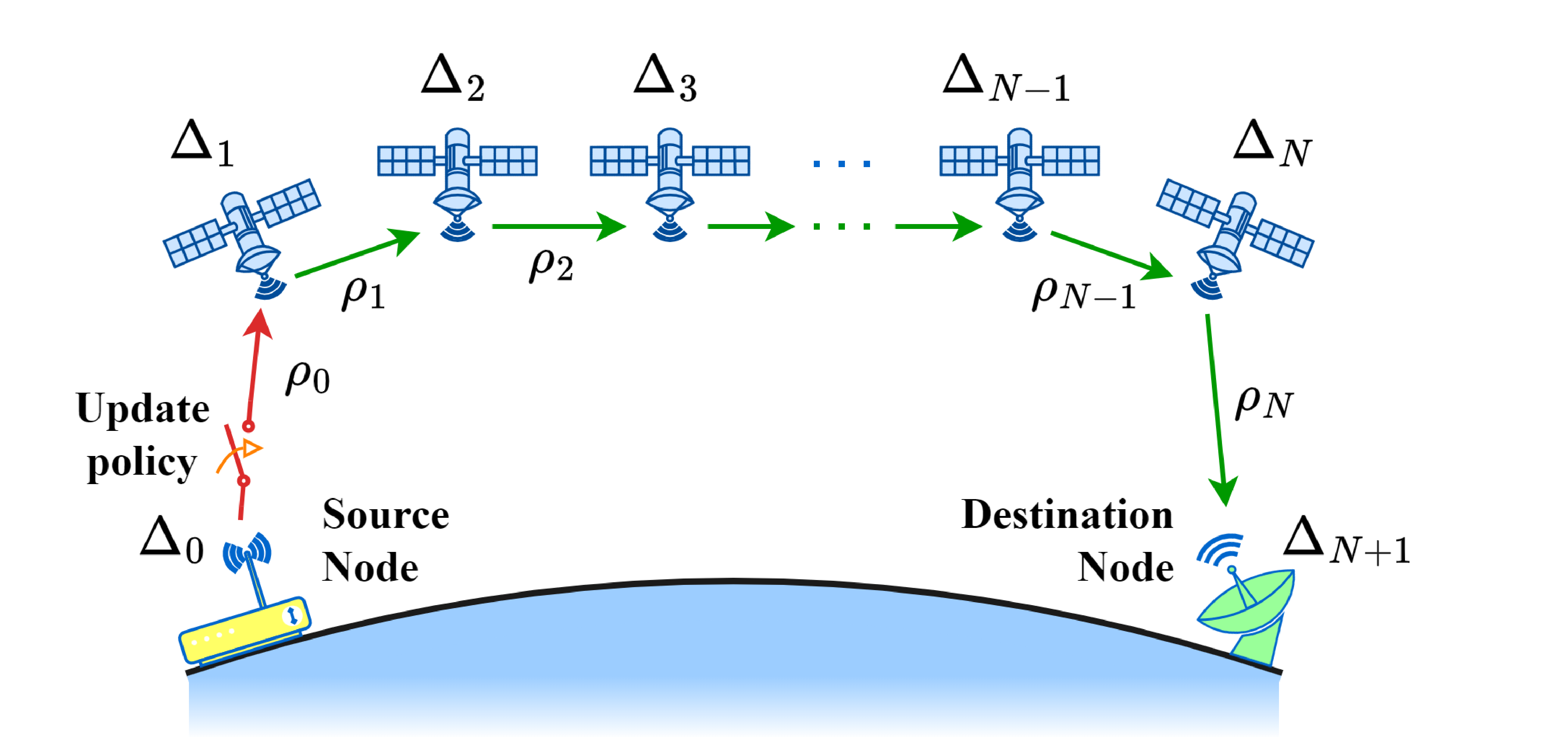} 
		\caption{Multi-hop communication via LEO satellites.}
		\label{fig_SatelliteSetup}
		\end{minipage}
        \hfill
            \begin{minipage}[c]{0.30\linewidth}
			\centering
		\includegraphics[scale=0.084, trim=150 0 125 0, clip]{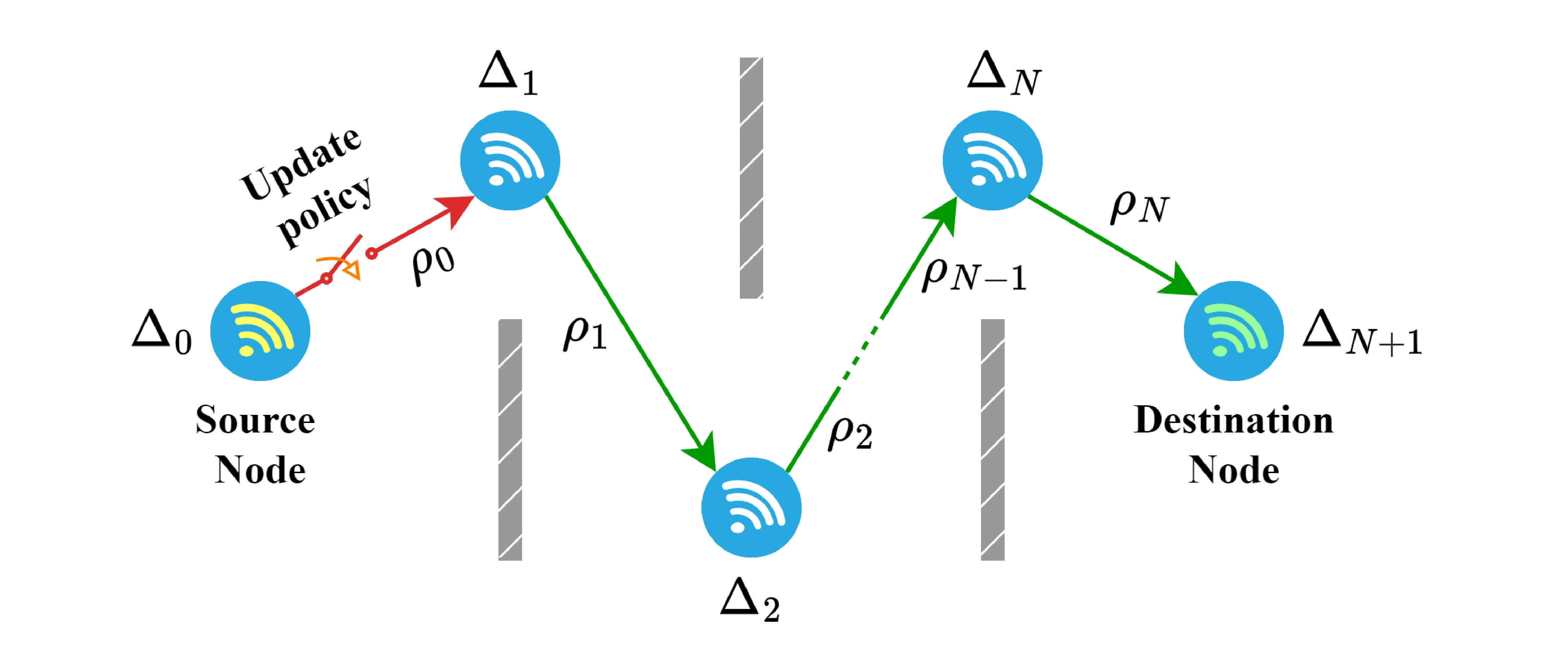} 
		\caption{Multi-hop communication in the presence of physical obstructions.}
		\label{fig_MeshSetup}
		\end{minipage}
		\hfill
        \vspace{-2em}
	\end{figure*}
	
    \section{Related Works}
    \label{sec_RelatedWorks}
	Several studies have investigated the distributions of AoI and Peak AoI (PAoI) in continuous-time systems using queueing theory \cite{costa2016age,champati2019distribution,inoue2019general,yates2020age,Chiariotti2021PAoI,jiang2021joint,abd2022closed,moltafet2022moment,fiems2023age,akar2025age,inoue2025characterizing}, while others have examined discrete-time settings \cite{kosta2021age,zhang2021age,akar2021discrete,ji2024age,Zhang2025AoIVehicles}. Notably, \cite{kosta2021age} derives general expressions for the stationary distributions and generating functions of AoI and PAoI in discrete-time single-server queues under various disciplines, together with methods for evaluating nonlinear age functions. Extending stochastic hybrid system techniques to discrete time, \cite{zhang2021age} models AoI and packet age as a two-dimensional Markov process in bufferless queues with Bernoulli arrivals. A matrix-analytic framework based on quasi-birth–death processes is proposed in \cite{akar2021discrete} to obtain exact per-source AoI and PAoI distributions in multi-source IoT systems with discrete phase-type service times. The study in \cite{ji2024age} investigates age-optimal packet scheduling with delayed feedback and long-term resource constraints, providing closed-form benchmarks for random and deterministic policies; and \cite{Zhang2025AoIVehicles} analyzes AoI and PAoI in multi-source Ber/Geo/1/1 systems under preemptive and non-preemptive policies, deriving closed-form expressions for distributions and averages.
	
	For content-based metrics, several works have examined the distribution of AoII \cite{maatouk2020age,chen2024minimizing,salimnejad2024age}. Specifically, \cite{maatouk2020age} derives stationary AoII distributions for symmetric multi-state Markov sources under always-update and threshold policies. The work in \cite{chen2024minimizing} investigates AoII in slotted systems with random transmission delays for two-state Markov sources under threshold-based updates. Using discrete-time Markov chain (\textsc{dtmc}) analysis, \cite{salimnejad2024age} provides stationary AoII and Age of Incorrect Version (AoIV) distributions for two-state Markovian sources under specified transmission policies. A content-based metric, the Age of Changed Information (AoCI), was introduced in \cite{wang2021age}, where the optimal threshold minimizing a weighted sum of AoCI and update costs under threshold-based policies was derived. Stationary distributions of VAoI have also been modeled for energy-harvesting systems \cite{delfani2024semantics} using \textsc{dtmc}s with stochastic energy arrivals and threshold-based transmissions. Furthermore, \cite{karevvanavar2024version} analyzes VAoI distributions in non-orthogonal multiple access fading broadcast channels with randomly arriving versioned packets and power constraints under a channel-only randomized stationary policy.
	
	Regarding multi-hop networks, \cite{ayan2020probability} derives the distribution of discrete-time AoI in $N$-hop systems with time-invariant packet loss through recursive formulations, while \cite{Vikhrova2020Mhop} studies continuous-time AoI and PAoI distributions in two-hop scenarios. However, most existing research on multi-hop networks, including \cite{Talak2018Mhop, Bedewy2019Mhop, buyukates2019age, Chiariotti2022Mhop, Tripathi2023Mhop, Kaswan2023Mhop, Sinha2024Tandem,asvadi2024age, Delfani2025LEO}, primarily focuses on average metrics.

	\section{System Model}
	\label{sec_SystemModel}
	We consider a communication network for transmitting data from a source to a destination node, which may be separated by either a single hop or multiple hops. The data at the source are constantly sampled from an information source and then transmitted according to an \emph{update policy}. The update policy schedules each data sample or \emph{update} to be either transmitted or skipped, while satisfying a constraint on the long-term average transmission rate. Specifically, the average update rate must not exceed a predefined limit. We assume a slotted time axis $t \in \{0,1,2,\cdots\}$, and our objective is to investigate the distribution of the VAoI at the destination node in both single-hop and multi-hop scenarios under various rate-constrained policies. The details of the system model are explained below. 
	
	\textit{Single- and Multi-hop Setups:} 
	We first consider a single-hop end-to-end status update system, as shown in Fig.~\ref{fig_SingleHopSetup}, where the destination node is one hop away and connected to the source through a direct but unreliable channel. The channel is modeled as an erasure channel, delivering each update with a success probability of $p_s$ per time slot. A reliable feedback channel from the destination to the source provides acknowledgments upon successful reception.
    We then extend this model to a multi-hop network comprising $N$ intermediate nodes that relay updates from the source to the final destination. This setup applies to communication networks such as Low Earth Orbit (LEO) satellite-based links between remote ground stations (Fig.~\ref{fig_SatelliteSetup}) and mesh networks where direct or line-of-sight connections are infeasible due to physical obstructions or link budget limitations (Fig.~\ref{fig_MeshSetup}). The relaying route is known a priori, and each node constantly forwards the most recent version until it is successfully received. Each link between nodes $i$ and $i+1$ is characterized by its own success probability $\rho_i$ for $i \in \{0,1,2,\dots,N\}$, where $\rho_0 = p_s$.
    In both setups, updates are transmitted at the beginning of each time slot and are received at its end; this sequence is preserved across all links. Each node stores only the latest version, discarding the previous ones. 
	
	\textit{Version Age of Information:} 
    We adopt the VAoI as the performance metric that captures both the timeliness and relevance of information. Unlike the AoI, which measures the time elapsed since the generation timestamp $u(t)$ of the freshest received update, $\Delta^{AoI}(t) \!=\! t \!-\! u(t),$ the VAoI quantifies how many versions the receiver lags behind the information source. At time slot $t$, the VAoI, denoted $\Delta(t)$, is defined as $\Delta(t) \!=\! V_S(t) \!-\! V_R(t)$, where $V_S(t)$ is the version index at the information source and $V_R(t)$ the version stored at the receiver. Fig. \ref{fig_VAoIAoI} shows the evolution of AoI and VAoI over discrete slots. A successful update occurs at $t \!=\! 3$, causing AoI and VAoI to drop to $1$ and $0$, respectively, since the update is one slot old and no new versions have been generated. Between $t \!=\! 3$ and the next update at $t \!=\! 12$, the AoI grows linearly to $9$, while the VAoI reaches $4$, reflecting the generation of $4$ new versions at the information source. 
    
	Fig. \ref{fig_VAoI_Evolution} illustrates VAoI dynamics in a two-hop network. The information source generates new versions at time slots $0$, $2$, $3$, $6$, and $8$, with its version index $V_S(t)$ increasing by one at the start of each subsequent slot. Node $0$ holds the most recent version available in the network; hence, its VAoI is always zero, i.e., $V_0(t) \!=\! V_S(t)$. The versions at node $1$, denoted by $V_1(t)$, are updated according to the update policy and the transmissions from node $0$. If a transmission at time $t$ succeeds, then $V_1(t\!+\!1) \!=\! V_0(t)$; otherwise, $V_1(t\!+\!1) \!=\! V_1(t)$. The corresponding VAoI, $\Delta_1(t) \!=\! V_S(t) \!-\! V_1(t)$, is listed in parentheses in the second row of the table and plotted in blue. Similarly, node $2$ stores versions $V_2(t)$ received from node $1$, which transmits in every slot.\footnote{In practice, if feedback channels are available between relay nodes, retransmissions of already delivered versions can be avoided; otherwise, relay nodes transmit continuously each time slot. In both cases, however, the presented VAoI analysis remains valid.} Here, $V_2(t\!+\!1) \!=\! V_1(t)$ if the transmission at $t$ succeeds; otherwise, $V_2(t\!+\!1) \!=\! V_2(t)$. The VAoI at node $2$, $\Delta_2(t) \!=\! V_S(t) \!-\! V_2(t)$, is shown in parentheses in the third row of the table and plotted in green in Fig. \ref{fig_VAoI_Evolution}.

     \begin{figure}[!t]
		\centering
		\includegraphics[scale=0.045]{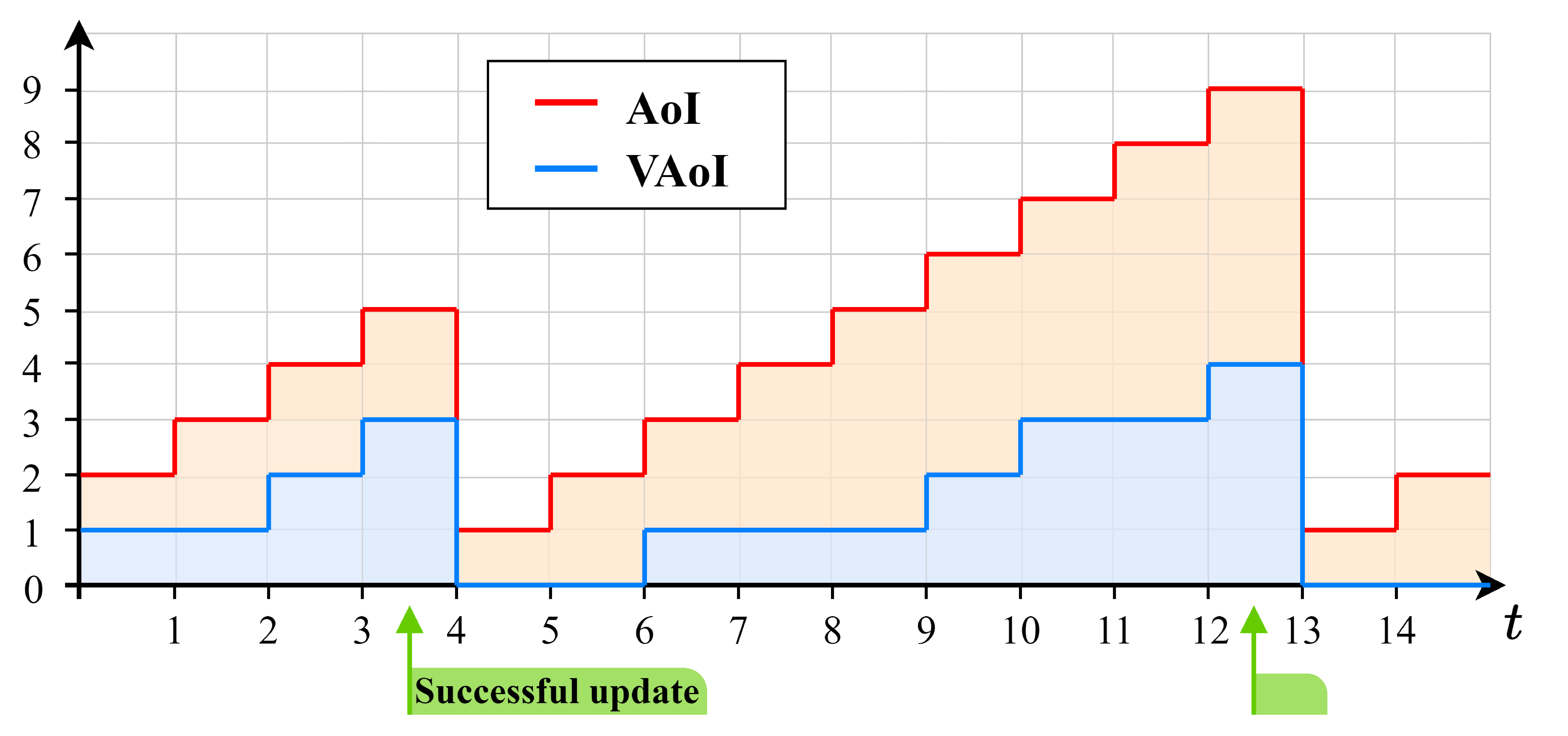} 
		\caption{Evolution of VAoI compared to AoI over time.}
		\label{fig_VAoIAoI}
	\end{figure}
    
	\begin{figure}[!t]
		\centering
		\includegraphics[scale=0.082]{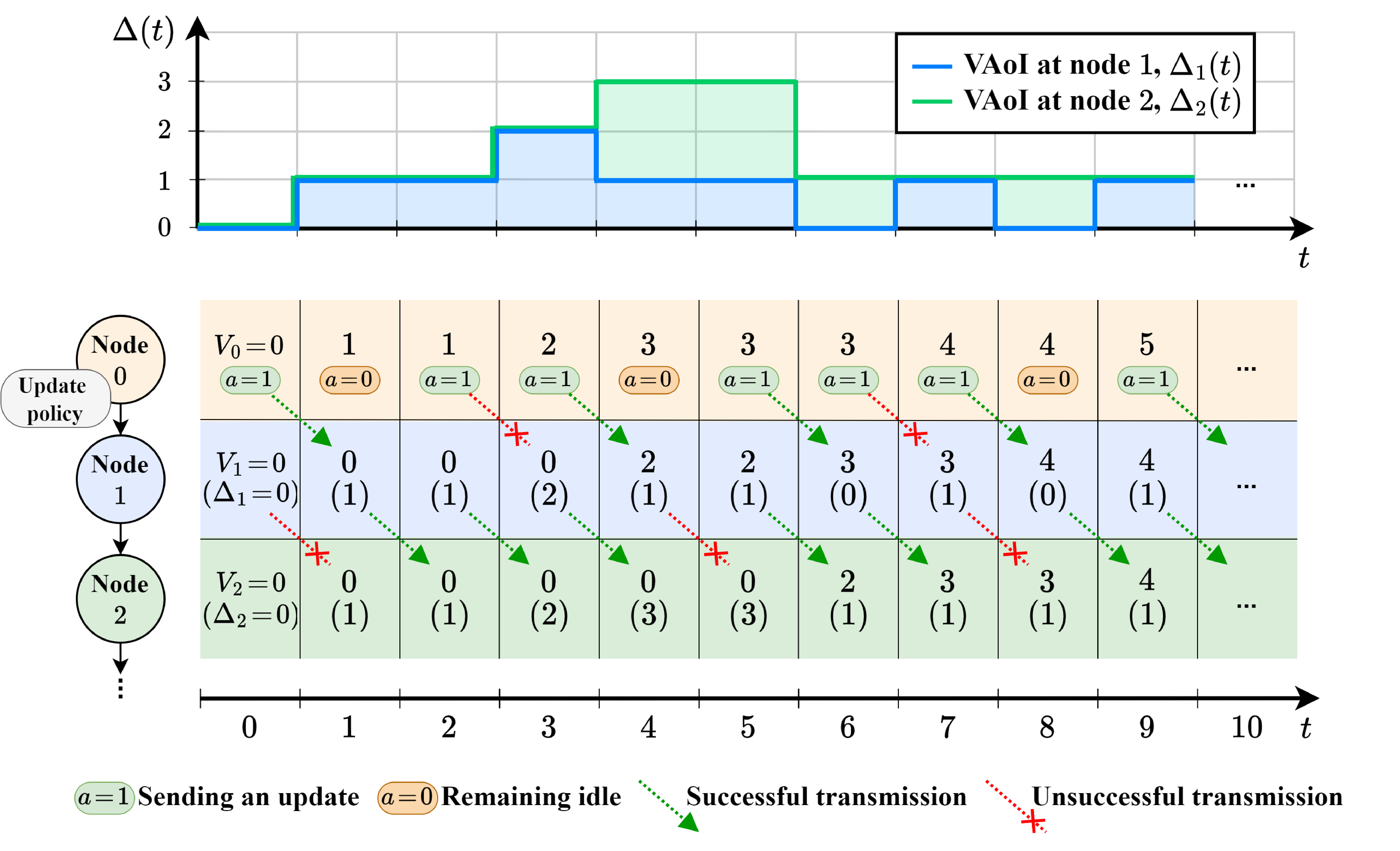} 
		\caption{Evolution of VAoI within the network over time.}
		\label{fig_VAoI_Evolution}
	\end{figure}
	
	In our system model, source versions are generated independently in each time slot with probability (w.p.) $p_g$, following a Bernoulli process.\footnote{This implies that the intervals between version generations follow a geometric distribution, the discrete-time analogue of the exponential distribution commonly used in continuous-time systems.} The evolution of versions at the destination node depends on the version generation process, the update policy, and the communication channel performance. 
	
	\textit{Remark:} All subsequent analysis for VAoI can be reduced to the discrete-time AoI by setting $p_g = 1$, i.e., when content changes are disregarded, and only data staleness due to elapsed time slots is considered. This demonstrates that \emph{VAoI is a more general metric, with AoI representing a special case.}
	
	\textit{Update Policies:}
	Various update policies can be considered to satisfy the average update rate constraint, ranging from myopic approaches—such as uniform rate transmission—to randomized stationary policies, as well as the optimal threshold policy derived from a Constrained Markov Decision Process (\textsc{cmdp}) problem (see Appendix \ref{Sec_CMDP}). Assuming a discrete-time, we represent the transmission \emph{action} at time $t$ under policy $\pi$ by the binary variable $a^\pi(t)$; $a^\pi(t)=1$ if a transmission is attempted at time $t$ and $0$ otherwise. The average update rate constraint can be expressed as:
	\begin{align}
		\lim_{T\rightarrow \infty} \frac{1}{T} \mathbb{E} \left [\sum_{t=0}^{T-1} a^\pi \!\left(t\right)\right ]  \leq \alpha, 
		\label{eq_constraint}
	\end{align}
	where $0 < \alpha \leq 1$ denotes the maximum average update rate. 
	We consider the following update policies in our analysis:
	\begin{itemize}[leftmargin=1em]
		\item \textbf{Randomized stationary policy}: In each time slot, a transmission occurs with probability $\alpha$; that is, $a^\pi(t) = 1$ with probability $\alpha$, and $a^\pi(t) = 0$ otherwise.
		\item \textbf{Uniform policy}: Transmissions occur periodically every $D$ samples, such that $a^\pi(t) = 1$ for $t \in \{0, D, 2D, \dots\}$, and $a^\pi(t) = 0$ otherwise.
		The maximum feasible value of $D$ that satisfies the constraint \eqref{eq_constraint} is given by $D = \lceil \frac{1}{\alpha} \rceil$, where $\lceil x \rceil$ denotes the smallest integer greater than or equal to $x$.
		
		\item \textbf{Threshold policy}: Transmission occurs only when the VAoI at the receiving node exceeds the threshold $\Delta_T$, that is, $a^\pi(t) = 1$ when $\Delta(t) \geq \Delta_T$ and $a^\pi(t) = 0$ otherwise. The smallest threshold $\Delta_T$ satisfying constraint \eqref{eq_constraint} is adopted. This policy is optimal for on-off scheduling, as demonstrated by the \textsc{cmdp} formulation and the proof provided in Appendix \ref{Sec_CMDP}.
	\end{itemize}

	\section{Analysis of VAoI in a Single-hop Setup}
	\label{Sec_OneHop}
	We analyze the VAoI in a single-hop system under the three aforementioned update policies using a \textsc{dtmc} model. The VAoI in the next time slot is stochastically determined by the current VAoI, the update policy, and system dynamics, including the version generation and channel success processes. 
	
	\textit{Balance equations in \textsc{dtmc}:} the steady-state probability of state $n$, i.e., the long-run probability of VAoI being equal to $n$ is given by:
    
	\begin{align}
		\label{eq_Balance}
		\mu_n = \sum_{j=0}^{\infty} P_{jn} \mu_j, \qquad n \in \{0,1,2,\dots\},
	\end{align}
    
	\noindent where $\sum_{n=0}^{\infty} \mu_n = 1$. Here, $P_{jn}$ denotes the transition probability from state $j$ to state $n$, i.e., $P_{jn}=\mathbb{P}\left( \Delta(t+1)=n \mid  \Delta(t)=j \right)$.
	The steady-state distribution of a \textsc{dtmc} exists if it is \emph{irreducible}, i.e., every state can be reached from every other state, and \emph{positive recurrent}, i.e., the expected return time to each state is finite~\cite[Sec. 1.8]{norris1998markov}. A \textsc{dtmc} is said to be \emph{ergodic} if it is irreducible, positive recurrent, and also \emph{aperiodic}, i.e., the chain does not become \emph{stuck} in a cycle of fixed length. For an ergodic \textsc{dtmc}, the long-term expected value and the expected value of the time average converge to the stationary mean~\cite[Sec. 1.10]{norris1998markov}:
	\begin{align}
		\label{eqn_MeanVAoI}
		\bar{\Delta} \!=\! \lim_{t \to \infty} \mathbb{E}[\Delta(t)] \!=\! \lim_{T \to \infty} \mathbb{E}\left[ \frac{1}{T} \sum_{t=0}^{T-1} \Delta(t) \right] \!=\! \sum_{n = 0}^{\infty} n \mu_n.
	\end{align}
	
	Unless stated otherwise, the Markov chains induced by the update policies in this section are ergodic. Since every VAoI state can reach state $0$ (or $1$) and vice versa with positive probability, and the transitions allow for self-resets and exits from any loop, the chain is irreducible, positive recurrent, and aperiodic. We proceed with analyzing their steady-state distributions.

	\subsection{VAoI of randomized stationary policy}
	\begin{proposition}
		\label{Prop_StateProbRandom}
		The steady-state probability of VAoI being in state $n$ under a randomized stationary policy with transmission probability $\alpha$ is given by:
		\begin{align}
			\mu_n = 
			\begin{cases}
				\frac{\alpha p_s (1-p_g)}{\beta}, & n=0,\\
				\frac{\alpha p_s p_g}{\beta^2}, & n=1,\\
				\left [ \frac{(1-\alpha p_s)p_g}{\beta} \right ]^{n-1} \! \mu_{1},& n \geq 2,
			\end{cases}
		\end{align}
		where $\beta=1-(1-\alpha p_s)(1-p_g)$.
	\end{proposition}
	\begin{proof}
		The proof is provided in Appendix \ref{Appen_Proof_RSpolicy_SSProbs}.
	\end{proof}
	
	\begin{lemma}
		\label{Lemma_AvgVAoIRandomized}
		The average VAoI under a randomized stationary policy with transmission probability $\alpha$ is given by: $\bar{\Delta} = \frac{p_g}{\alpha p_s}$.
	\end{lemma}
	
	\begin{proof}
		The recurrence relation of $\mu_n$ for $n \geq 1$ is geometric: $\mu_n = r^{n-1} \mu_1$, where $r = \frac{(1 - \alpha p_s)p_g}{\beta}$. Thus, the expected steady-state value in \eqref{eqn_MeanVAoI} is given by $\bar{\Delta} = \frac{\mu_1}{(1 - r)^2}$. Noting that $\mu_1 = \frac{\alpha p_s p_g}{\beta^2}$, the average VAoI is obtained as $\bar{\Delta} = \frac{p_g}{\alpha p_s}$.
	\end{proof}

	\subsection{VAoI of uniform policy}
	
	Under the uniform policy, the Markov chain’s transition matrix is not stationary; rather, it evolves periodically with period $D$. Consequently, a steady-state distribution does not exist. However, this periodicity facilitates the analysis of the long-term average proportion of time spent in state $n$:
	\begin{align}
		\mu_n = \frac{1}{D} \sum_{q=1}^{D} \mu_n^{(q)},\qquad \hfill n \in \{0,1,2,\cdots\},
	\end{align}
	
	Here, $\mu_n^{(q)}$ denotes the steady-state probability of the time-homogeneous Markov chain: $Y^{(q)}(k) = \Delta(kD + q)$, which samples the original VAoI process at time indices spaced at intervals of $D$, starting from phase offset $q \in \{1,2,\dots,D\} $. This $\mu_n$ represents the long-run time average (or steady-state occupancy probability) for the original periodically time-inhomogeneous Markov process~\cite[Sec. 1.8, Theorem 1.8.5]{norris1998markov}.
	
	\begin{proposition}
		\label{Prop_StateProbUniform}
		The steady-state occupancy probability of VAoI being in state $n$ under a uniform policy with transmission interval $D$ is given by $\mu_n = \frac{1}{D} \sum_{q=1}^{D} \mu_n^{(q)}$, where:
		\begin{align}
			\mu^{(q)}_n = 
			\begin{cases}
				\frac{p_sb^{q}_0}{\beta}, & n=0, \\
				\frac{1}{\beta} \left [ (1-p_s) \sum_{i=1}^{n} b^{D}_{i}\mu_{n-i}^{(q)} + p_sb^{q}_n \right],& 1 \leq n \leq q, \\
				\frac{1}{\beta} \left [ (1-p_s) \sum_{i=1}^{n} b^{D}_{i}\mu_{n-i}^{(q)} \right],& q < n \leq D, \\
				\frac{1}{\beta} \left [ (1-p_s) \sum_{i=1}^{D} b^{D}_{i}\mu_{n-i}^{(q)} \right],& n \geq D+1,
			\end{cases}
		\end{align}
		with $b^{q}_z=\binom{q}{z}p_g^z (1-p_g)^{q-z}$ and $\beta=1-(1-p_s)b^D_0$.
	\end{proposition}
	
	\begin{proof}
		The proof is provided in Appendix \ref{Appen_Proof_UniformPolicy_SSProbs}.
	\end{proof}
	
	Proposition \ref{Prop_StateProbUniform} provides recursive equations for calculating the stationary distribution of the VAoI under the uniform policy. Using these probabilities, the average VAoI $\bar{\Delta}$ can be calculated via \eqref{eqn_MeanVAoI}.

	\subsection{VAoI of threshold policy}
	\begin{proposition}
		\label{Prop_StateProbThr}
		The steady-state probability of VAoI being in state $n$ under a threshold policy with threshold $\Delta_T$ is given by:
		\begin{itemize}
			\item For $\Delta_T \in \{0,1\}$: The same as the randomized stationary policy with $\alpha = 1$, as presented in Proposition \ref{Prop_StateProbRandom}.
			\item For $\Delta_T \geq 2$:
			\begin{align}
				\mu_n =
				\begin{cases}
					\frac{p_s(1-p_g)}{(\Delta_T-1)p_s+\beta}, & n=0,\\
					\frac{p_s}{(\Delta_T-1)p_s+\beta}, & 1 \leq n \leq \Delta_T\!-\!1, \\
					\frac{p_g}{\beta} \mu_{\Delta_T-1}, & n=\Delta_T,\\
					\left [\frac{(1-p_s)p_g}{\beta} \right]^{n-\Delta_T} \mu_{\Delta_T},& n \geq \Delta_T \!+\! 1,
				\end{cases}
			\end{align}
		\end{itemize}
		where $\beta=1-(1-p_s)(1-p_g)$.
	\end{proposition}
	
	\begin{proof}
		The proof is provided in Appendix \ref{Appen_Proof_ThrPolicy_SSProbs}.
	\end{proof}
	
	\begin{lemma}
		\label{Lemma_AvgVAoIThreshold}
		The average VAoI under a threshold policy with threshold $\Delta_T$ is given by:
		\begin{align}
			\label{eqn_AvgVAoIThreshold}
			\bar{\Delta}_{(\Delta_T)} = \frac{1}{2} \frac{(\Delta_T\!-\!1)\Delta_T p_s}{(\Delta_T\!-\!1)p_s + \beta} + \frac{p_g}{p_s},
		\end{align}
		where $\beta=1-(1-p_s)(1-p_g)$.
	\end{lemma}
	
	\begin{proof}
		
		For $\Delta_T = 0$ and $\Delta_T = 1$, the average VAoI is identical, as shown in the proof of Proposition \ref{Prop_StateProbThr}, and equals the average VAoI in Lemma \ref{Lemma_AvgVAoIRandomized} with $\alpha = 1$, i.e., $\bar{\Delta}_{(\Delta_T=0)} \!=\! \bar{\Delta}_{(\Delta_T=1)} \!=\! \frac{p_g}{p_s}$. For $\Delta_T \geq 2$, using the steady-state probabilities from Proposition \ref{Prop_StateProbThr}, the expected VAoI, $\bar{\Delta}_{(\Delta_T)} \!=\! \sum_{n=0}^\infty n \mu_n$, is given by:
		\begin{align}
			\bar{\Delta}_{(\Delta_T)} 
			&\!\!=\! \mu_{\Delta_T-1} \Bigg\{\!\! \sum_{n=1}^{\Delta_T-1} \!\! n
			\!+\! \frac{p_g}{\beta} \Delta_T \!+\! 
			\! \frac{p_g}{\beta} \left[ \! \frac{r}{(1\!-\!r)^2} \!+\! \Delta_T \frac{r}{1\!-\!r} \!\right]  \!\! \Bigg\}, \notag
		\end{align}
		where $r = \frac{(1-p_s)p_g}{\beta}$, and after some algebraic manipulation, the final expression for the average VAoI \eqref{eqn_AvgVAoIThreshold} is obtained.
	\end{proof}
	
	\begin{theorem}
		\label{Theorem_OptimalThreshold}
		The optimal threshold-based policy minimizing the average VAoI under the rate constraint \eqref{eq_constraint} is a randomized mixture of two threshold policies with thresholds $\Delta_T^\ast$ and $\Delta_T^\ast - 1$, applied with probabilities $\gamma$ and $1 - \gamma$, respectively. The optimal threshold is:
        \vspace{-4pt}
		\begin{align}
			\label{eqn_OptimalThreshold}
			\Delta_T^\ast= \Bigg\lceil \frac{p_g}{p_s}\left( \frac{1}{\alpha} - 1 +p_s\right) \!\! \Bigg\rceil,
		\end{align}
		and the corresponding mixing probability $\gamma$ is:
		\begin{align}
			\label{Optimal_gamma}
			\gamma = \frac{R(\Delta_T^\ast \!-\! 1) - \alpha}{R(\Delta_T^\ast \!-\! 1) - R(\Delta_T^\ast)},
		\end{align}
		where $R(\Delta_T) = \frac{p_g}{(\Delta_T-1)p_s+\beta}$ for $\Delta_T \geq 1$, and $R(0) = 1$.
	\end{theorem}
	
	\begin{proof}
		The proof is provided in Appendix \ref{Appen_Proof_OptimalThr}.
	\end{proof}
	
	The resulting optimal average VAoI under the \emph{mixed threshold policy} is given by:
	$\bar{\Delta}^\ast = \gamma \bar{\Delta}_{(\Delta_T^\ast)} + \left(1-\gamma\right) \bar{\Delta}_{(\Delta_T^\ast-1)}$,
	where $\bar{\Delta}_{(\Delta_T^\ast)}$ and $\bar{\Delta}_{(\Delta_T^\ast-1)}$ are obtained from \eqref{eqn_AvgVAoIThreshold}.
    
    \emph{Remark:} In a highly constrained system where $\alpha$ is sufficiently small, the optimal threshold \eqref{eqn_OptimalThreshold} grows to $\lceil \frac{p_g}{\alpha p_s} \rceil$, and the resulting average VAoI \eqref{eqn_AvgVAoIThreshold} approaches $\frac{p_g}{2\alpha p_s}$, which is half that of the randomized stationary policy, $\frac{p_g}{\alpha p_s}$.

    \begin{table*}[!t]
		\centering
		\caption{Average VAoI at node $N\!+\!1$ for various update policies.}
		\begin{tabular}{|>{\centering\arraybackslash}m{1.7cm}|>{\centering\arraybackslash}m{5.95cm}|>{\centering\arraybackslash}m{3.95cm}|>{\centering\arraybackslash}m{4.65cm}|}
			\hline
            \rule{0pt}{2ex}
			 Update policy & {Randomized stationary (transmission probability $\alpha$}) & {Uniform (transmission interval $D$)} & {Optimal threshold-based (threshold $\Delta_T$)}
			\\
			\hline
            \rule{0pt}{3ex}
            $\bar{\Delta}_{N+1}$ & 
			$\frac{p_g}{\alpha p_s} + p_g \sum_{i=1}^{N} \frac{1}{\rho_i}$ & 
			$\sum_{n = 0}^{\infty} n \mu_n + p_g \sum_{i=1}^{N} \frac{1}{\rho_i}$ & 
			$\gamma \bar{\Delta}_{(\!\Delta_T^\ast\!)} \!+\! \left(1\!-\!\gamma\right) \bar{\Delta}_{(\!\Delta_T^\ast-1\!)} \!+\! p_g \! \sum_{i=1}^{N} \!\! \frac{1}{\rho_i}$ 
			\\
			\hline
		\end{tabular}
		\label{tab_AvgVAoI_allPolicies}
        \vspace{-1.5em}
	\end{table*}
	
	\section{Analysis of VAoI in a Multi-hop Setup}
	\label{Sec_MultiHop}
	
	We first demonstrate that the VAoI at each node can be expressed in terms of the VAoI at the preceding node. 
	
	\begin{proposition}
		\label{Prop_VAoInodei}
		The VAoI at node $i\!+\!1$ is given by:
		\begin{align}
			\label{eqn_VAoInodei}
			\Delta_{i+1}(t) &= \Delta_{i}(t-m_i) + \eta_{m_i}, \quad i=1,2,\dots,N,
		\end{align}
		where $m_i$ is a Geometric Random Variable (RV) with parameter $\rho_i$, and  $\eta_k$, for a given $k$, is a Binomial RV with parameters $k$ and $p_g$, for $i \in \{1, 2, \dots, N\}$ and $k \in \{0, 1, 2, \dots\}$:
		\begin{align}
			\mathbb{P}\!\left(m_i\!=\!\mathscr{\ell}\right)&\!=\!(1\!-\!\rho_i)^{\mathscr{\ell}-1}\rho_i, \quad \hfill \mathscr{\ell}=1,2,\dots. \label{eq_GeometricPMF} \\
            \mathbb{P}\!\left(\eta_k\!=\!r \!\mid\! k \right)&\!=\!\binom{k}{r}p_g^{r}(1-p_g)^{k-r},\quad r=0,1,\dots,k. \label{eq_BinomialPMF} 
		\end{align}
	\end{proposition}
	
	\begin{proof}
		The proof is provided in Appendix \ref{Appen_Proof_VAoInodei}.
	\end{proof}
	
	\begin{lemma}
		\label{Lemma_VAoIDestNode}
		The VAoI at the destination node is given by:
		\begin{align}
			\label{eq_VAoInodeNp1}
			\Delta_{N+1}(t)=\Delta_{1} (t-\tau_{N})+\beta_{N},
		\end{align}
		where $\tau_{N} = \sum_{i=1}^{N} \!m_{i}$ and $\beta_{N}=\sum_{i=1}^{N} \!\eta_{m_i}$ are two RVs with expected values 
		$\mathbb{E} \left[\tau_N \right] = \sum_{i=1}^{N} \frac{1}{\rho_i}$ and $\mathbb{E} \left[\beta_N \right] = p_g \sum_{i=1}^{N} \frac{1}{\rho_i}$.
	\end{lemma}
	\begin{proof}
		The proof is provided in Appendix \ref{Appen_Proof_VAoIDestNode}.
	\end{proof}
	\vspace{-2pt}
	The variable $\tau_N$ is the \emph{relaying delay} of each version from node $1$ to node $N\!+\!1$ through $N$ relaying nodes, while $\beta_N$ represents the number of version generations at the source during this delay. Specifically, $\tau_N$ and $\beta_N$ are the sums of independent Geometric and Binomial RVs, $\{m_i\}_{i=1}^{N}$ and $\{\eta_{m_i} \! \mid \! m_i\}_{i=1}^{N}$, respectively. Their probability mass functions (PMFs) can be derived by convolving the PMFs of the individual components. This derivation simplifies in two cases: (1) when $\{\rho_i\}_{i=1}^{N} \!=\! \rho$, $\tau_N$ follows a Negative Binomial distribution, $\tau_N \!\sim\! NegBin(N, \rho)$, representing the number of trials required to achieve the $N$-th success, with $\mathbb{P}(\tau_N \!=\! \ell) \!=\! \binom{\ell \!-\! 1}{N \!-\! 1} \rho^N (1 \!-\! \rho)^{\ell \!-\! N}$ for $\ell \!=\! N, N \!+\! 1, \dots$; and (2) when $N$ is large, by the Central Limit Theorem, $\tau_N \!\sim\! \mathcal{N} \!\left(\sum_{i=1}^{N} \!\!\frac{1}{\rho_i}, \, \sum_{i=1}^{N} \!\!\frac{1 - \rho_i}{\rho_i^2}\right)$. Moreover, since $\{m_i\}_{i=1}^N$ are independent, $\beta_N \!\mid\! \tau_N \sim Bin(\tau_N, p_g)$ always holds.
	
	\begin{theorem}
		\label{Lemma_AvgVAoIlastNode}
		The average VAoI of the receiver node which is $N+1$ hops away from the source is given by:
		\begin{align}
			\label{eqn_AvgVAoIlastNode}
			\bar{\Delta}_{N+1} = \bar{\Delta}_1 + p_g \sum_{i=1}^{N} \frac{1}{\rho_i}. 
		\end{align}
	\end{theorem}
	
	\begin{proof}
		The proof is provided in Appendix \ref{Appen_Proof_AvgVAoIlastNode}.
	\end{proof}
	
	Using the single-hop update policies from Section \ref{Sec_OneHop} and substituting the respective average VAoI expressions for $\Delta_1$ in \eqref{eqn_AvgVAoIlastNode}, the average VAoI at the final node $N+1$ is presented in Table \ref{tab_AvgVAoI_allPolicies}. In this table, $\mu_n$ for the uniform policy follows from Proposition \ref{Prop_StateProbUniform}, while $\bar{\Delta}_{(\Delta_T)}$, $\Delta_T^\ast$, and $\gamma$ for the optimal threshold policy are defined in \eqref{eqn_AvgVAoIThreshold}–\eqref{Optimal_gamma}.

	\vspace{-1pt}
	\section{Numerical Results}
	\label{Sec_NumericResults}
	We evaluate the analytical results, starting with the single-hop setup under various update policies and extending the analysis to the multi-hop setup. Simulations were conducted over $10^4$ time slots, and the results were averaged over $400$ Monte Carlo iterations to obtain the steady-state values.

    \vspace{-3pt}
	\subsection{Single-hop Setup}
	We numerically validate the analytical results for the stationary distribution and the average VAoI presented in Section~\ref{Sec_OneHop}. Specifically, we compute the fraction of time slots in which $\Delta_1$ (or $\Delta$ in Section~\ref{Sec_OneHop}) equals $n$, for $n \in \{0,1,2,\ldots\}$, under randomized stationary, uniform, and threshold policies. Parameters are set to $p_s = 0.8$ and $p_g = 0.3$. The simulation results perfectly match and validate the analytical results in Propositions~\ref{Prop_StateProbRandom}, \ref{Prop_StateProbUniform}, and \ref{Prop_StateProbThr}, as illustrated in Fig.~\ref{fig_SS_probs_D4} for $\alpha = 0.25$.
	
	For further evaluation, Fig. \ref{fig_SS_probs_D20} shows the stationary distributions for $\alpha = 0.05$, corresponding to a more constrained average update rate.\footnote{Simulation curves are omitted from the remaining results to eliminate redundancy and enhance clarity, as they replicate the analytical results exactly.} Under stricter rate constraints, the VAoI distribution becomes more dispersed, with higher VAoI values occurring with greater probability. The randomized stationary and uniform policies produce smooth distributions with longer tails, whereas the threshold policy exhibits an almost \emph{uniform} distribution with considerably shorter tails. As stated in Proposition \ref{Prop_StateProbThr}, the VAoI distribution $\mu_n$ under the threshold policy is uniform for $1 \!\leq\! n \!\leq\! \Delta_T - 1$, drops by a factor of $\frac{p_g}{\beta}$ at $n = \Delta_T$, and then decays exponentially at rate $(1 - p_s) p_g/\beta$. 
	
	Fig. \ref{fig_AvgVAoI_alpha} presents the average VAoI under different policies. The threshold policy consistently delivers the best performance, while the randomized stationary policy performs worse than the uniform policy. The advantage of the threshold policy lies in its ability to keep VAoI values low and mitigate the occurrence of higher ones. Since the average VAoI is a weighted sum of steady-state probabilities, with larger values contributing more (see Equation~\eqref{eqn_MeanVAoI}), reducing the probability of high VAoI significantly improves performance. This explains why the uniform policy outperforms the randomized stationary policy, which exhibits a higher probability of large VAoI values (Figs.~\ref{fig_SS_probs_D4} and \ref{fig_SS_probs_D20}).

    \begin{figure*}[!t]
		\centering
		\begin{minipage}[b]{0.325\linewidth}
			\centering
			\includegraphics[scale=0.42]{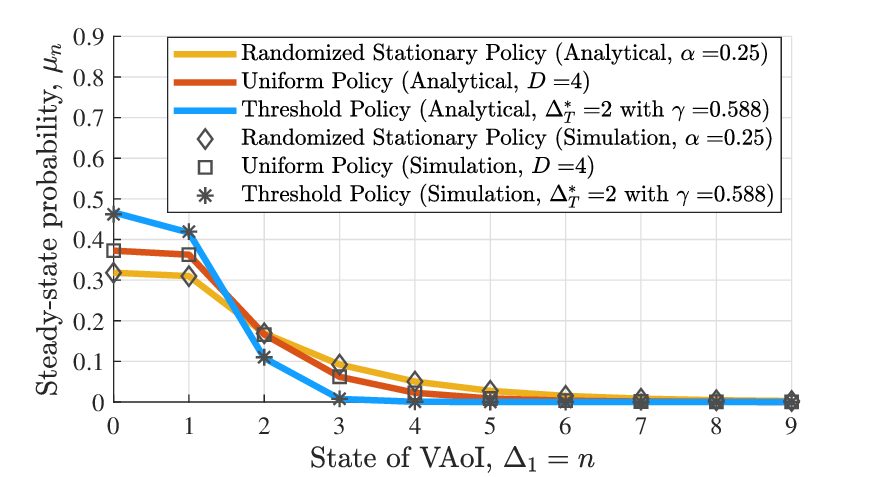} 
			\caption{Stationary distribution of VAoI at node $1$ under various policies for $\alpha = 0.25$.}
			\label{fig_SS_probs_D4}
		\end{minipage}
		\hfill
		\begin{minipage}[b]{0.325\linewidth}
			\centering
			\includegraphics[scale=0.42]{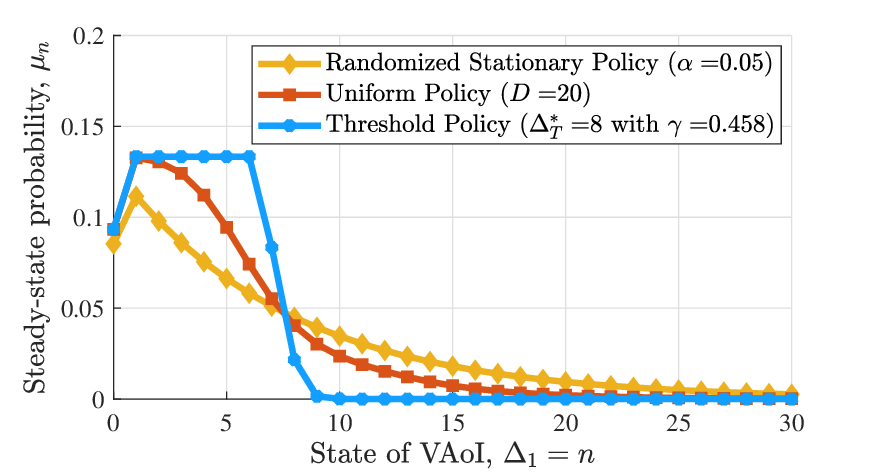} 
			\caption{Stationary distribution of VAoI at node $1$ under various policies for $\alpha = 0.05$.}
			\label{fig_SS_probs_D20}
		\end{minipage}
		\hfill
		\begin{minipage}[b]{0.325\linewidth}
			\centering
			\includegraphics[scale=0.41]{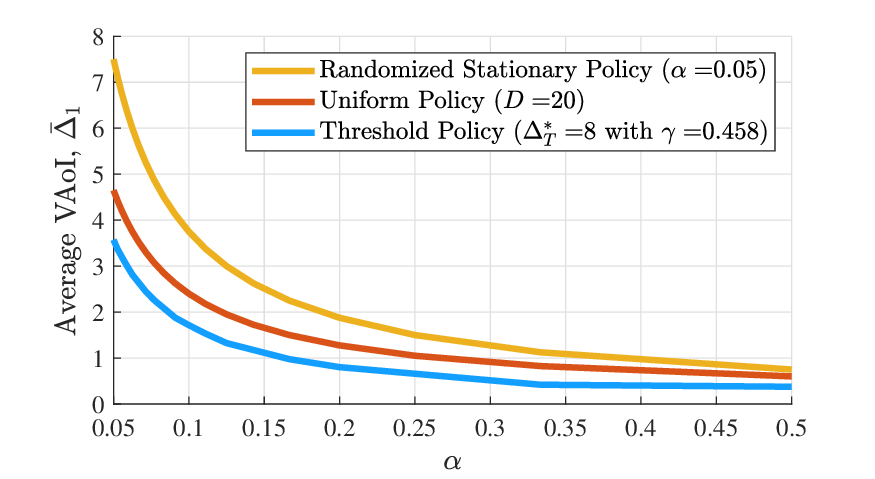} 
			\caption{Average VAoI at node $1$ under various policies for $\alpha = 0.05$.}
			\label{fig_AvgVAoI_alpha}
		\end{minipage}
        \vspace{-1.5em}
	\end{figure*}
	
	Fig. \ref{fig_AvgVAoI_alpha} illustrates that, to maintain a target VAoI of, for instance, $2$, the uniform policy reduces the required update rate by $35\%$ (from $0.188$ to $0.121$), while the optimal threshold policy achieves a $54\%$ reduction (to $0.086$), compared to the randomized policy. These results highlight that \emph{optimal VAoI-aware policies can significantly reduce the transmission rate without compromising the conveyed information, thereby improving energy efficiency in communication networks.}
	
	Fig. \ref{fig_Thresholds_ps_pg} illustrates the contours of the optimal threshold $\Delta_T^\ast$ and the heatmap of the average VAoI under the threshold policy across varying success probabilities $p_s$ and version generation probabilities $p_g$, for $\alpha\!=\!0.05$. Higher $p_g$ and lower $p_s$ lead to a larger average VAoI and a higher optimal threshold. As established in Theorem \ref{Theorem_OptimalThreshold}, the average VAoI increases with $\Delta_T$, where $\Delta_T^\ast$ is given by Equation~\eqref{eqn_OptimalThreshold}. Under stringent rate constraints (i.e., very low $\alpha$), $\Delta_T^\ast$ can be approximated as $\lceil \frac{p_g}{\alpha p_s} \rceil$. Moreover, when $\alpha \!\geq\! \frac{p_g}{\beta}$, we have $\Delta_T^\ast \!=\! 1$, meaning updates occur at all non-zero VAoI states, thereby attaining the minimum achievable average VAoI, $\bar{\Delta} \!=\! \frac{p_g}{p_s}$.

    \begin{figure*}[!t]
		\centering
		\begin{minipage}[b]{0.325\linewidth}
			\centering
			\includegraphics[scale=0.42]{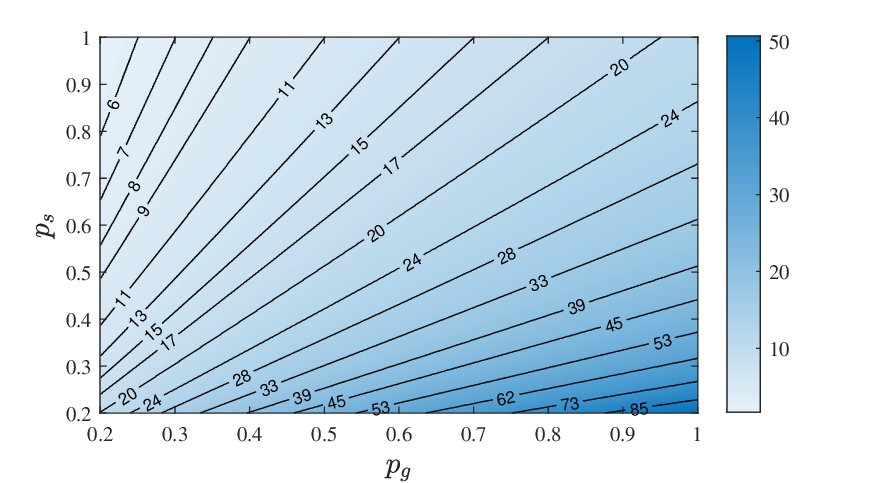} 
			\caption{Contour plot of $\Delta_T^\ast$ and heatmap of $\bar{\Delta}_1$ for the optimal threshold policy versus $(p_g, p_s)$.}
			\label{fig_Thresholds_ps_pg}
		\end{minipage}
		\hfill
		\begin{minipage}[b]{0.325\linewidth}
			\centering
			\includegraphics[scale=0.42]{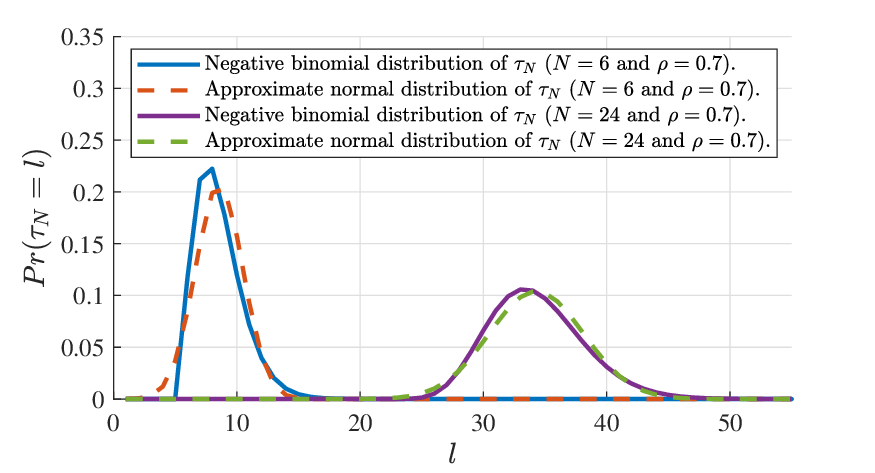} 
			\caption{Distribution of the relaying delay for each version, $\tau_N$.}
			\label{fig_Tau_N_Distribution}
		\end{minipage}
		\hfill
		\begin{minipage}[b]{0.325\linewidth}
			\centering
			\includegraphics[scale=0.39]{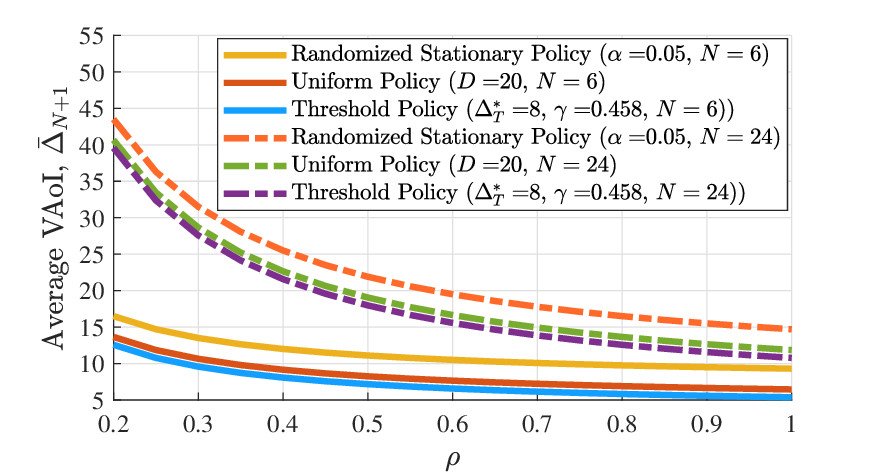} 
			\caption{Average VAoI at node $N\!+\!1$ under various policies ($\alpha = 0.05$, $N=6$ and $24$).}
			\label{fig_AvgVAoI_Np1_rho}
		\end{minipage}
        \vspace{-1.5em}
	\end{figure*}
	
	\subsection{Multi-hop Setup}
	According to Lemma \ref{Lemma_VAoIDestNode}, the VAoI at the destination node $N+1$ in a multi-hop setup with $N$ relays equals the average VAoI at node $1$, shifted by the relaying delay $\tau_N$, plus the number of version generations at the source during this delay, denoted by the Binomial RV $\beta_N$. Consequently, the long-term average VAoI, $\bar{\Delta}_{N+1}$, equals the average VAoI at node $1$, $\bar{\Delta}_1$, plus an offset corresponding to the expected number of version generations during relaying, given by $p_g \sum_{i=1}^{N} \frac{1}{\rho_i}$. The relaying delay $\tau_N$ follows a negative Binomial distribution under ${\rho_i}_{i=1}^N = \rho$. As $N$ increases, this distribution converges to a normal distribution, as illustrated in Fig.~\ref{fig_Tau_N_Distribution} for $\rho = 0.7$ with $N = 6$ and $24$.
	The average VAoI, $\bar{\Delta}_{N+1}$, is depicted in Fig. \ref{fig_AvgVAoI_Np1_rho} for various policies as a function of $\rho$, with parameters $p_s=0.8$, $p_g=0.3$, $\alpha=0.05$, for $N = 6$ and $24$. The average VAoI, $\bar{\Delta}_{N+1}$, exhibits a linear increase with respect to $N$ and a reciprocal polynomial decrease with respect to $\rho$, such that $\bar{\Delta} \propto \frac{N}{\rho}$. That is, as the network size increases, maintaining the average VAoI at the destination requires the incorporation of more reliable links.

	\section{Conclusion}
	\label{sec_Conclusion}
	We analyzed the VAoI in both single-hop and multi-hop networks. Closed-form expressions for the stationary distribution of the VAoI, along with its average, were derived for several rate-constrained transmission policies, including randomized stationary, uniform, and threshold-based. Furthermore, we obtained explicit formulas for the optimal threshold and the resulting optimal VAoI under the threshold policy. Compared to other policies, the optimal threshold policy achieved the lowest average VAoI or maintained a similar VAoI with significantly fewer transmissions.

	\appendices
	
	\section{Optimal On-Off Policy for the CMDP Problem}
	\label{Sec_CMDP}

	Consider an on--off transmission policy $\pi$ that, at each time slot $t$, decides whether to transmit ($a^\pi(t) = 1$) or remain idle ($a^\pi(t) = 0$), i.e., $\pi = \big(a^\pi(0), a^\pi(1), \dots \big)$, to minimize the time-average VAoI under an average update rate constraint, formulated as a \textsc{cmdp}:
    
	\begin{align}
		\label{eq_CMDP}
		\min_{\pi \in \Pi} \ &\lim_{T\rightarrow\infty} {\frac{1}{T} \mathbb{E} \left[ \sum_{t=0}^{T-1} \Delta(t) \Big| s(0) \right]}, \qquad \hfill s.t.\ \eqref{eq_constraint},
	\end{align}
    
    \noindent where $\Pi$ denotes all feasible policies.
    The \textsc{cmdp} is characterized by state $s(t) \in S$, action $a(t) \in A = \{0,1\}$, transition probability $\mathbb{P}(s(t+1) \mid s(t), a(t))$, and transition cost $C(s(t), a(t), s(t+1))$, where, for brevity, the superscript $\pi$ is omitted. The state, $s(t) = \Delta(t)$, denotes the VAoI at the receiver; imposing an upper bound $\Delta_{\text{max}}$ yields a finite state space $S = \{0, 1, \dots, \Delta_{\text{max}}\}$. The transition probabilities are:
	\begin{align}
		\label{eqn_CMDPTransProbs}
		\mathbb{P}\!\left(\Delta^\prime | \Delta,a\right)
		\!=\! 
        \scalebox{0.9}{$
		\begin{cases}
			p_g & a\!=\!0,\  \Delta^\prime\!=\!\Delta\!+\!1,\ \Delta\!\!<\!\!\Delta_{\text{max}},\\
			\bar{p}_g & a\!=\!0,\  \Delta^\prime\!=\!\Delta,\ \Delta\!\!<\!\!\Delta_{\text{max}},\\
			1 & a\!=\!0,\  \Delta^\prime\!=\!\Delta\!=\!\Delta_{\text{max}},\\
			\bar{p}_sp_g & a\!=\!1,\  \Delta^\prime\!=\!\Delta\!+\!1,\ \Delta\!\!<\!\!\Delta_{\text{max}},\\
			\bar{p}_s \bar{p}_g & a\!=\!1,\  \Delta^\prime\!=\!\Delta,\ \Delta\!\!<\!\!\Delta_{\text{max}},\\
			\bar{p}_s & a\!=\!1,\ \Delta^\prime\!=\!\Delta\!=\!\Delta_{\text{max}},\\
			p_sp_g & a\!=\!1,\  \Delta^\prime\!=\!1,\\
			p_s \bar{p}_g & a\!=\!1,\  \Delta^\prime\!=\!0,
		\end{cases} 
        $}
	\end{align}
    where $\bar{p}_s \!=\! 1 \!-\! p_s$ and $\bar{p}_g \!=\! 1 \!-\! p_g$. The transition cost at state $s(t)$ under action $a(t)$ is defined as the resulting VAoI, i.e., $C\!\left(s(t), a(t), s(t\!+\!1)\right) \!=\! \Delta(t\!+\!1)$.
	The primal \textsc{cmdp} problem \eqref{eq_CMDP} can be reformulated as a Lagrangian dual problem by introducing a multiplier $\lambda \geq 0$:
	\begin{align}
		\label{eq_DualMDP}
		\sup_{\lambda \geq 0} {\min_{\pi \in \Pi} \mathcal{L}(\lambda,\pi)},
	\end{align}
	where $\mathcal{L}(\lambda,\pi)$ denotes the Lagrangian function:
	\begin{align}
		\label{eq_LagFunction}
		\mathcal{L}(\lambda,\pi) \!=\! \lim_{T\rightarrow\infty} {\frac{1}{T} \mathbb{E} \!\left[ \sum_{t=0}^{T-1} \left \{ \Delta(t) \!+\! \lambda a(t) \right \} \Big| s(0) \right] \!-\! \lambda \alpha}.
	\end{align}
	
	Here, $g(\lambda) = \mathcal{L}(\lambda, \pi_\lambda^*)$ denotes the dual function, and $\pi_\lambda^*$ is the policy minimizing $\mathcal{L}(\lambda, \pi)$ for fixed $\lambda$:
	\begin{align}
		\label{eqn_DualMDPpolicy}
		\pi^\ast_\lambda = \argmin_{\pi \in \Pi} \lim_{T\rightarrow\infty} { \frac{1}{T} \mathbb{E} \left[ \sum_{t=0}^{T-1} \left \{ \Delta(t) + \lambda a(t) \right \} \Big| s(0) \right] }.
	\end{align}
	This corresponds to solving an unconstrained \textsc{mdp} with a modified transition cost:
	\begin{align}
		\label{eq_LagrangianCostFunction}
		C_\lambda \! \left(s(t),a(t),s(t+1)\right) = \Delta(t+1) + \lambda a(t).
	\end{align}
	
	For a finite state space $S$, the growth condition in \cite[Eq.~11.21]{altman1999constrained} holds. Since the transition cost $C(s(t), a(t), s(t+1)) \geq 0$ is bounded below, the conditions of \cite[Corollary~12.2]{altman1999constrained} are satisfied, ensuring the optimal solutions of the dual and primal problems coincide. Thus, the optimal solution to the primal \textsc{cmdp} \eqref{eq_CMDP} is found by solving $\sup_{\lambda \geq 0} g(\lambda)$, where $\pi_{\lambda}^\ast$ comes from \eqref{eqn_DualMDPpolicy}. Specifically, the optimal policy is obtained by first solving the unconstrained \textsc{mdp} \eqref{eqn_DualMDPpolicy} for fixed $\lambda$ to get $\pi_{\lambda}^\ast$, and then optimizing $\lambda$ as in \eqref{eq_DualMDP}. We proceed to prove that $\pi_{\lambda}^\ast$ is a threshold policy.
	
	\begin{proposition}
		The optimal policy of the MDP problem \eqref{eqn_DualMDPpolicy} is a threshold policy.
	\end{proposition}
	
	\begin{proof}
		We begin by establishing that the \textsc{mdp} is weakly accessible, thereby ensuring the existence of an optimal policy. An \textsc{mdp} is weakly accessible if its state space can be partitioned into a transient set $S_t$ and a communicating set $S_c$, where all states in $S_c$ are mutually reachable under some stationary policy. For any stationary stochastic policy $\pi$ assigning positive probability to each action $a \in \{0,1\}$, any state $\Delta'$ is reachable from $\Delta$. Specifically, if $\Delta' < \Delta$, take $a=1$ once, then $a=0$ for $\Delta'$ steps; if $\Delta' \geq \Delta$, take $a=0$ for $\Delta' - \Delta$ steps. Hence, the \textsc{mdp} is weakly accessible; thereby by Proposition 4.2.3 in \cite{bertsekas2011dynamic}, the optimal average cost $J_\lambda^*$ is independent of the initial state $s(0)$. Proposition 4.2.6 guarantees the existence of an optimal policy $\pi_\lambda^*$, and Proposition 4.2.1 ensures $J_\lambda^*$, the value function $\mathbb{V}(s)$, and $\pi_\lambda^*$ satisfy the Bellman equations:
		\begin{align}
			\label{eqn_Bellman}
			J_\lambda^*\!\!+\!\mathbb{V}(s)\!=\!\!\!\min_{a\in\left\{0,1\right\}}\!{\!Q_\lambda(s,a)}, \ \  
			\pi^\ast(s) \!\in\!  \argmin_{a\in\left\{0,1\right\}}{Q_\lambda(s,a)},
		\end{align}
		where $Q_\lambda(s,a) = C_\lambda(s,a)+\!\sum_{s^\prime\in S}{\mathbb{P}\!\left(s^\prime \big | s,a\right)\!\mathbb{V}(s^\prime)}.$
		Here, $C_\lambda(s,a)$ represents the average cost per slot, defined by the transition costs as: $C_\lambda(s,a) = \sum_{s^\prime\in S} {\mathbb{P}\!\left(s^\prime \big | s,a\right) C_\lambda\left(s,a,s^\prime\right)},$ with $C_\lambda(s,a,s^\prime) = \Delta^\prime + \lambda a$.
		The Bellman equation for state $s = \Delta$ can be written as $a^\ast(\Delta) = 1$ if $Q_\lambda(\Delta, 1) < Q_\lambda(\Delta, 0)$, and $a^*(\Delta) = 0$ otherwise. Thus, the optimal action $a^*(\Delta)$ depends on the sign of the difference $\text{d}\mathbb{V}(\Delta) = Q_\lambda(\Delta, 1) - Q_\lambda(\Delta, 0)$. 
		We next show that $\text{d} \mathbb{V}(\Delta)$ is a decreasing function of $\Delta$. For $\Delta\!^- \leq \Delta\!^+$, we prove that $\text{d} \mathbb{V}(\Delta\!^+) \leq \text{d} \mathbb{V}(\Delta\!^-)$, i.e., 
		\begin{align}
			\label{eq_DVineq}
			\text{d} \mathbb{V}(\Delta\!^+) - \text{d} \mathbb{V}(\Delta\!^-) \leq 0.
		\end{align}
		
		This implies a threshold policy: if $\text{d} \mathbb{V}(\Delta_T) < 0$ for some $\Delta_T$, then $\text{d} \mathbb{V}(\Delta) < 0$ for all $\Delta \geq \Delta_T$, so the optimal action remains $1$ for all such states.		
		Using \eqref{eqn_CMDPTransProbs}, we obtain:
		\begin{align*}
			\text{d} \mathbb{V}(\Delta\!^+) &\!-\! \text{d} \mathbb{V}(\Delta\!^-) 
			\!=\! {-p_s \big[ \Delta\!^+ \!-\! \Delta\!^- \big]} \!-\! \bar{p}_g p_s \big[ \mathbb{V}(\Delta\!^+) \!-\! \mathbb{V}(\Delta\!^-)\big] \\
			& \!-\!\bar{p}_g p_s \big[ \mathbb{V}(\Delta\!^+ \!+\! 1) \!-\! \mathbb{V}(\Delta\!^- \!+\! 1)\big].
		\end{align*}
		
		The first term is non-positive. Thus, to prove inequality \eqref{eq_DVineq}, it suffices to prove that $\mathbb{V}(\Delta)$ is increasing in $\Delta$, i.e., for $\Delta\!^- \leq \Delta\!^+$, 
		$\mathbb{V}(\Delta\!^-) \leq \mathbb{V}(\Delta\!^+).$
		We use the Value Iteration Algorithm (\textsc{via}) and induction. \textsc{via} converges to $V(\Delta)$ regardless of the initial $\mathbb{V}_0(\Delta)$, i.e., $\lim_{k\to\infty} \mathbb{V}_k(\Delta) = \mathbb{V}(\Delta)$ for all $\Delta \in S$.
		The \textsc{via} iteration is:
		\begin{align*}
			\mathbb{V}_{k+1}(\Delta)\!=\!\!\min_{a\in\left\{0,1\right\}} \bigg\{\!\underbrace{\sum_{\Delta^\prime\in S} \mathbb{P}\!\left(\Delta^\prime \big|\Delta,a\right) \Big( \Delta^\prime \!+\! \lambda a \!+\! \mathbb{V}_k(\Delta^\prime) \Big)}_{Q_{\lambda,k}(\Delta,a)}\!\bigg\}.
		\end{align*}
		
		We prove by induction that $\mathbb{V}_k(\Delta\!^-) \!\leq\! \mathbb{V}_k(\Delta\!^+)$ for all $k \!\geq\! 0$. For $k\!=\!0$, $\mathbb{V}_0(\Delta)\!=\!0$, so the claim holds. Assume $\mathbb{V}_k(\Delta\!^-) \!\leq\! \mathbb{V}_k(\Delta\!^+)$ and show $\mathbb{V}_{k+1}(\Delta\!^-) \!\leq\! \mathbb{V}_{k+1}(\Delta\!^+)$. Since $\mathbb{V}_{k+1}(\Delta) = \min \{\mathbb{V}_{k+1}^0(\Delta), \mathbb{V}_{k+1}^1(\Delta)\}$ with $\mathbb{V}_{k+1}^0 \!=\! Q_{\lambda,k}(\Delta,0)$ and $\mathbb{V}_{k+1}^1 \!=\! Q_{\lambda,k}(\Delta,1)$, it suffices to prove $\mathbb{V}_{k+1}^0(\Delta\!^-) \!\leq\! \mathbb{V}_{k+1}^0(\Delta\!^+)$ and $\mathbb{V}_{k+1}^1(\Delta\!^-) \!\leq\! \mathbb{V}_{k+1}^1(\Delta\!^+)$, since then $\min\{\mathbb{V}_{k+1}^0(\Delta\!^-),\mathbb{V}_{k+1}^1(\Delta\!^-)\} \!\leq\! \min\{\mathbb{V}_{k+1}^0(\Delta\!^+),\mathbb{V}_{k+1}^1(\Delta\!^+)\}$. By the induction hypothesis, all bracketed terms below are non-positive:
		\begin{align*}
			\mathbb{V}&^0_{\!k+1}(\Delta\!^-) \!-\! \mathbb{V}^0_{\!k+1}(\Delta\!^+) \!=\! \big[\Delta\!^- \!-\! \Delta\!^+ \big] \\ 
			&\!+\!  \bar{p}_g \big[ \mathbb{V}_{\!k}(\Delta\!^-) \!-\! \mathbb{V}_{\!k}(\Delta\!^+) \big] \!+\! p_g \big[\mathbb{V}_{\!k}(\Delta\!^- \!\!+\! 1) \!-\! \mathbb{V}_{\!k}(\Delta\!^+ \!\!+\! 1) \big].  \\
			\mathbb{V}&^1_{\!k+1}(\Delta\!^-) \!-\! \mathbb{V}^1_{\!k+1}(\Delta\!^+) \!=\! \bar{p}_s\big[\Delta\!^- \!\!-\! \Delta\!^+ \big] \\ 
			&\!+\!  \bar{p}_g \bar{p}_s \big[ \mathbb{V}_{\!k}(\Delta\!^-) \!-\! \mathbb{V}_{\!k}(\Delta\!^+) \big] \!+\! p_g \bar{p}_s \big[\mathbb{V}_{\!k}(\Delta\!^- \!\!+\! 1) \!-\! \mathbb{V}_{\!k}(\Delta\!^+ \!\!+\! 1) \big]. 
		\end{align*}
        Thus $\mathbb{V}_{k+1}(\Delta\!^-) \!\leq\! \mathbb{V}_{k+1}(\Delta\!^+)$, completing the proof.
	\end{proof}

	\section{}
	\label{Appen_Proof_RSpolicy_SSProbs}
	\begin{proof}[Proof of Proposition \ref{Prop_StateProbRandom}]
		At each time slot, a transmission attempt occurs w.p. $\alpha$ and succeeds w.p. $p_s$, giving a success probability of $\alpha p_s$. Upon success, the VAoI resets to $0$ or $1$ depending on whether a new version is generated in that slot. Otherwise, if no transmission occurs or it fails (w.p. $1 - \alpha p_s$), the VAoI increases by $0$ or $1$, again depending on version generation. The Markov chain with transition probabilities is shown in Fig.~\ref{fig_MC_RSpolicy}, and solving its balance equations as described below yields Proposition \ref{Prop_StateProbRandom}.
		\[
			\begin{array}{ll}
				\mu_0 \!=\! (1\!-\!p_g)\mu_0 \!+\! \alpha p_s (1\!-\!p_g) {\sum_{i=1}^{\infty} \! \mu_i}, \\
				\mu_1 \!=\! p_g \mu_0 \!+\! \big[ \alpha p_s p_g \!+\! (1\!-\!\alpha p_s)(1\!-\!p_g) \big] \mu_1 \!+\! \alpha p_s p_g {{\sum_{i=2}^{\infty} \! \mu_i}}, \\
				\mu_n \!=\! (1\!-\!\alpha p_s)p_g \mu_{n-1} \!+\! (1\!-\!\alpha p_s)(1\!-\!p_g)\mu_n,\quad \hfill n \geq 2. \qedhere 
			\end{array}
		\]
	\end{proof}
	
	\begin{figure}[tb]
		\centering
		\includegraphics[scale=0.071]{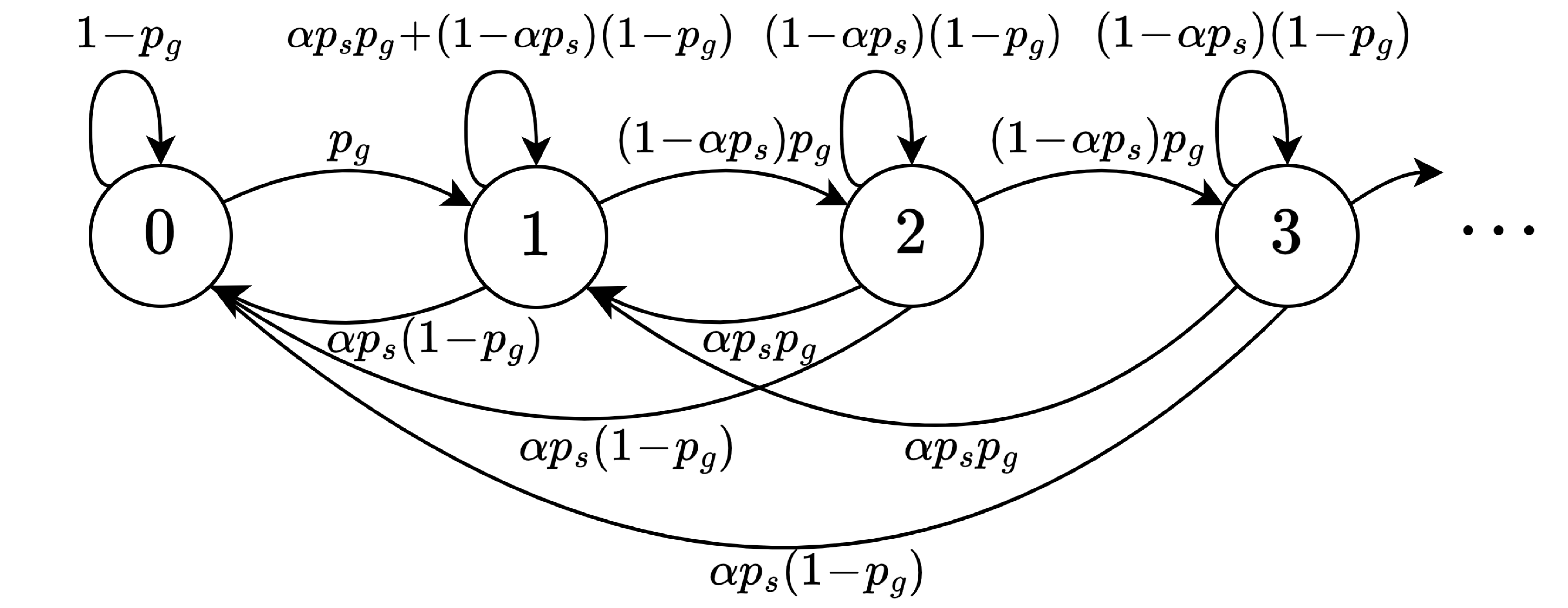} 
        \vspace{2pt}
		\caption{\textsc{dtmc} model of the VAoI under a randomized stationary policy.}
        \vspace{2pt}
		\label{fig_MC_RSpolicy}
	\end{figure}

    \begin{figure*}[!b]
		\centering
        \vspace{-0.8em}
		\includegraphics[scale=0.070, trim=140 15 140 15,clip]{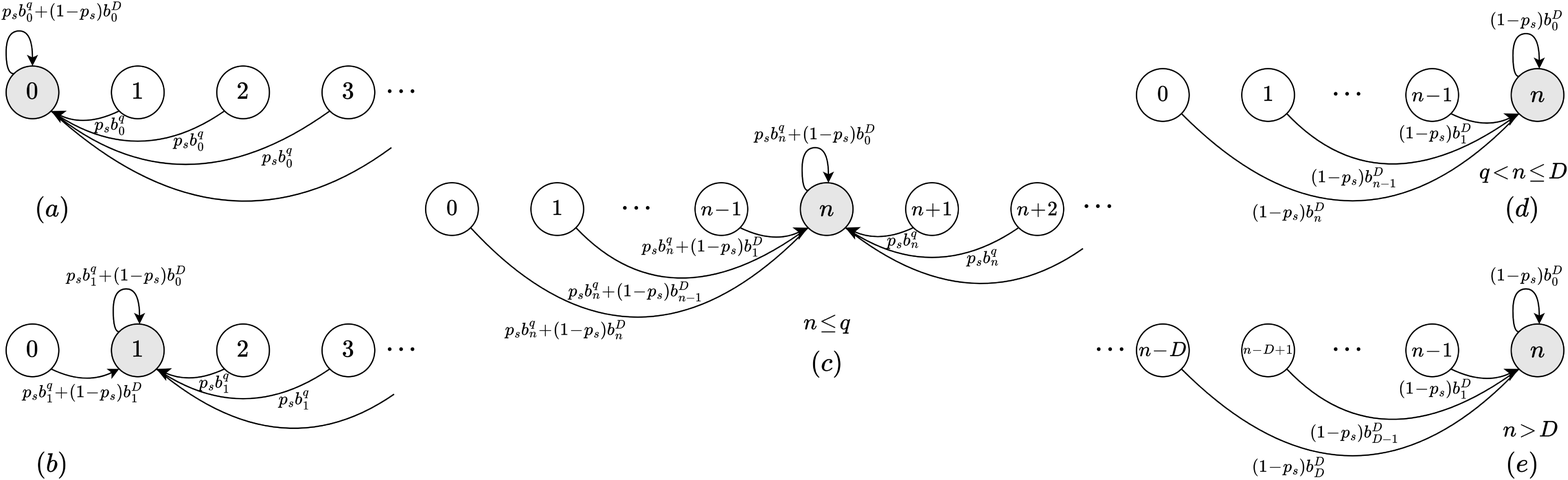} 
		\vspace{4pt}
		\caption{State transitions to target states $(a)\ 0$, $(b)\ 1$, $(c)\ 1 \!<\! n \!\leq\! q$, $(d)\ q \!<\! n \!\leq\! D$, and $(e)\ n \!>\! D$ in the Markov chain under a uniform policy.}
		\label{fig_UniformPolicy_all}
	\end{figure*}

	\section{}
	\label{Appen_Proof_UniformPolicy_SSProbs}
	
	\begin{proof}[Proof of Proposition \ref{Prop_StateProbUniform}]
		For each sampled Markov chain $Y^{(q)}(k) = \Delta(kD \!+\! q)$, the VAoI transitions from $Y^{(q)}(k)$ to $Y^{(q)}(k\!+\!1)$ depending on whether the transmission at $t \!=\! (k\!+\!1)D$ succeeds (or not) and on the number of source versions generated in the $q$ (or $D$) slots.  
		\begin{itemize}[leftmargin=*]
			\item \textit{Successful transmission} (w.p. $p_s$): the VAoI resets to zero and subsequently increases with the versions generated within $q$ time slots: it equals $0$ if none are generated (w.p. $p_s b^{q}_0$); $1$ if one is generated (w.p. $p_s b^{q}_1$); and so on, up to $n \leq q$, where $n$ versions are generated (w.p. $p_s b^{q}_n$).
			Version generation per slot follows i.i.d. Bernoulli with parameter $p_g$; thus, generating $z$ versions in $q$ slots follows $Bin(q, p_g)$ w.p. $b^{q}_z \!=\! \binom{q}{z} p_g^z (1 - p_g)^{q-z}$.
			
			\item \textit{Unsuccessful transmission} (w.p. $1 - p_s$): the VAoI increases by the number of versions generated during $D$ slots: $0$ w.p. $(1 - p_s) b^{D}_0$, $1$ w.p. $(1 - p_s) b^{D}_1$, and so on, up to $D$ w.p. $(1 - p_s) b^{D}_D$.
		\end{itemize}
		
		Fig. \ref{fig_UniformPolicy_all} illustrates the Markov chain states and their transitions to target states—$0$, $1$, $1 \!<\! n \!\leq\! q$, $q \!<\! n \!\leq\! D$, and $n \!>\! D$—together with the associated transition probabilities. By formulating the balance equations for each state, we derive:
		\[
			\begin{array}{ll}
				\mu^{(q)}_0 \!=\! (1\!-\!p_s) b^{D}_0\mu^{(q)}_0 \!+\! p_sb^{q}_0 {\sum_{i=0}^{\infty} \! \mu^{(q)}_i}, \\
				\mu^{(q)}_1 \!=\! (1\!-\!p_s) \left( b^{D}_1\mu^{(q)}_0 \!+\! b^{D}_0\mu^{(q)}_1 \right) \!+\! p_sb^{q}_1 {\sum_{i=0}^{\infty} \! \mu^{(q)}_i}, \\
				\mu^{(q)}_n \!=\! (1\!-\!p_s)\sum_{i=0}^{n} b^{D}_{i}\mu^{(q)}_{n-i} \!+\! p_sb^{q}_n {\sum_{i=0}^{\infty} \! \mu^{(q)}_i},\ \ \hfill 1 < n \leq q,\\
				\mu^{(q)}_n \!=\! (1\!-\!p_s)\sum_{i=0}^{n} b^{D}_{i}\mu^{(q)}_{n-i},\ \ \hfill q < n \leq D,\\
				\mu^{(q)}_n \!=\! (1\!-\!p_s)\sum_{i=0}^{D} b^{D}_{i}\mu^{(q)}_{n-i},\ \ \hfill n \geq D+1,
			\end{array}
		\]
		where, given that $\sum_{i=0}^{\infty} \mu^{(q)}_i = 1$, the first equation immediately gives $\mu^{(q)}_0 = \frac{p_sb^{q}_0}{\beta}$, while the other equations follow from straightforward algebraic manipulation.
	\end{proof}

	\section{}
	\label{Appen_Proof_ThrPolicy_SSProbs}
	
	\begin{proof}[Proof of Proposition \ref{Prop_StateProbThr}]
		For $\Delta_T = 0$, when the VAoI is zero, transmission decisions do not affect state transitions, since the VAoI depends solely on whether a new version is generated. Hence, the system dynamics and steady-state distribution are identical for $\Delta_T = 0$ and $\Delta_T = 1$. The case $\Delta_T = 0$ represents an \emph{always-update} policy, equivalent to a randomized stationary policy with $\alpha = 1$.
		For $\Delta_T \geq 2$, if $\Delta < \Delta_T$, no transmission occurs; the VAoI increments by $1$ if a new version is generated, otherwise it remains the same. If $\Delta \geq \Delta_T$, transmission occurs: on success, VAoI resets to $1$ or $0$ depending on whether a new version is generated; on failure, it increases by $1$ if a new version is generated, or remains unchanged otherwise. The resulting Markov chain is depicted in Fig.~\ref{fig_MC_ThrPolicyDTo2}, with balance equations provided for $\Delta_T \geq 2$:
		\[
			\begin{array}{ll}
				\mu_0 \!=\! (1\!-\!p_g)\mu_0 \!+\! p_s (1\!-\!p_g){\sum_{i=\Delta_T}^{\infty} \! \mu_i}, \\
				\mu_1 \!=\! p_g \mu_0 \!+\! (1\!-\!p_s) \mu_1 \!+\! p_s p_g {\sum_{i=\Delta_T}^{\infty} \! \mu_i} \!=\! \frac{\mu_0}{1\!-\!p_g}, \\
				\mu_n \!=\! p_g \mu_{n-1} \!+\! (1\!-\!p_g) \mu_{n} \!=\! \mu_{n-1}, \hfill 2 \leq n \leq \Delta_T\!-\!1, \\
				\mu_{\Delta_T} \!=\! p_g \mu_{\Delta_T-1} \!+\! (1\!-\!p_s)(1\!-\!p_g) \mu_{\Delta_T}, \\
				\mu_n \!=\! (1\!-\!p_s)p_g \mu_{n-1} \!+\! (1\!-\!p_s)(1\!-\!p_g) \mu_{n}, \quad \hfill n \geq \Delta_T\!+\!1,
			\end{array}
		\]
        
		The third line is omitted when $\Delta_T \!=\! 2$. By applying $\sum_{i=0}^{\infty} \mu_i \!=\! 1$  and simplifying, Proposition~\ref{Prop_StateProbThr} follows.
	\end{proof}
    
        \begin{figure}[tb]
			\centering
			\includegraphics[scale=0.07]{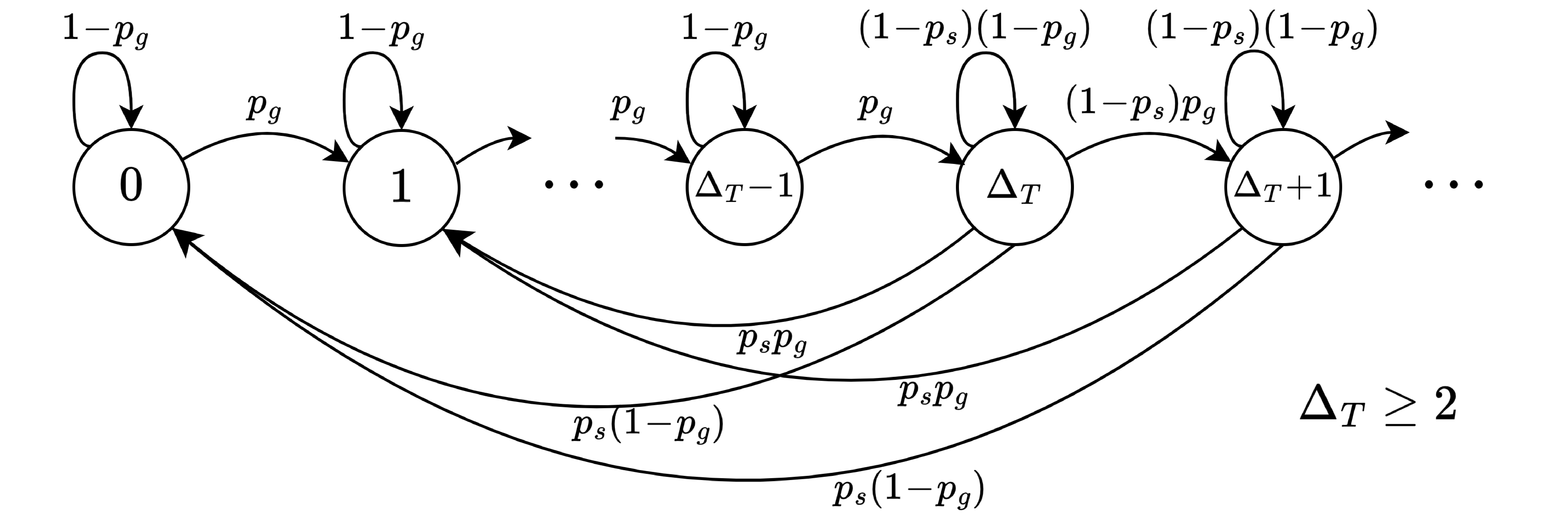} 
            \vspace{2pt}
			\caption{\textsc{dtmc} model of the VAoI under the threshold policy with $\Delta_T \!\geq\! 2$.}
            \vspace{2pt}
			\label{fig_MC_ThrPolicyDTo2}
		\end{figure}

	\section{}
	\label{Appen_Proof_OptimalThr}
	
	\begin{proof}[Proof of Theorem \ref{Theorem_OptimalThreshold}]
		Threshold policies that satisfy the average update rate constraint \eqref{eq_constraint} are considered feasible. For a threshold policy with parameter $\Delta_T$, the left-hand side of inequality \eqref{eq_constraint} represents the probability of the states that trigger transmission, i.e., $R(\Delta_T) \!=\! \mathbb{P}(\Delta \!\geq\! \Delta_T)$, given by:
		\begin{align}
			\label{eqn_FeasibleThresholds}
			R(\Delta_T) \!=\!\!\sum_{n=\Delta_T}^{\infty} \!\! \mu_n 
			\!\overset{(a)}{=}\!
			\begin{cases}
				1, & \Delta_T \!=\! 0, \\
				\frac{p_g}{(\Delta_T\!-\!1)p_s+\beta}, & \Delta_T \!\geq\! 1,
			\end{cases}
		\end{align}
		Step $(a)$ directly follows from the balance equation and the steady-state probabilities presented in Proposition \ref{Prop_StateProbThr}, utilizing a geometric series analysis similar to that in the proof of Lemma \ref{Lemma_AvgVAoIThreshold}. Therefore, a threshold policy is feasible if it satisfies $R(\Delta_T) \leq \alpha$, which can be simplified using \eqref{eqn_FeasibleThresholds}:
		\begin{align}
			\label{eqn_ThresholdLowerBound}
			\Delta_T \geq \frac{p_g}{p_s}\left( \frac{1}{\alpha} - 1 +p_s\right)\!.
		\end{align}
		
		This implies that $\Delta_T$ must exceed a certain lower bound. Meanwhile, Lemma~\ref{Lemma_AvgVAoIThreshold} shows that the average VAoI under the threshold policy \eqref{eqn_AvgVAoIThreshold} is an increasing function of $\Delta_T$, since it can be rewritten as:
		$\bar{\Delta} \!=\! \frac{\Delta_T}{2} \left( 1\!-\!\frac{\beta}{(\Delta_T\!-\!1)p_s \!+\! \beta} \right) \!+\! \frac{p_g}{p_s}.$
		The constant term $\frac{p_g}{p_s}$ remains fixed; the first term increases with $\Delta_T$ due to the linear growth of $\frac{\Delta_T}{2}$ and the rising value of $\left( 1 \!-\! \frac{\beta}{(\Delta_T - 1)p_s + \beta} \right)$, since $\frac{\beta}{(\Delta_T - 1)p_s + \beta}$ decreases as $\Delta_T$ grows. Therefore, minimizing the average VAoI under the update rate constraint entails selecting the smallest integer $\Delta_T$ that satisfies \eqref{eqn_ThresholdLowerBound}. However, exact equality may not always be attainable, as $\Delta_T$ must be integer-valued.
		To achieve $R(\Delta_T) \!=\! \alpha$, a randomized mixture policy can be employed~\cite{beutler1985optimal}\cite[Sec.~6.3]{altman1999constrained}, combining two thresholds $\Delta_T^\ast$ and $\Delta_T^\ast \!-\! 1$, where $R(\Delta_T^\ast) \!\leq\! \alpha$ and $R(\Delta_T^\ast \!-\! 1) \!>\! \alpha$, resulting in \eqref{eqn_OptimalThreshold}. The threshold $\Delta_T^\ast$ is applied w.p. $\gamma$ and $\Delta_T^\ast \!-\! 1$ with $1 \!-\! \gamma$, where $\gamma R(\Delta_T^\ast) + (1-\gamma) R(\Delta_T^\ast-1) = \alpha$, leading to \eqref{Optimal_gamma}. This \emph{mixed threshold policy} constitutes the optimal solution to the \textsc{cmdp} under the average rate constraint \eqref{eq_constraint} \cite{beutler1985optimal}.
	\end{proof}

	\section{}
	\label{Appen_Proof_VAoInodei}
	
	\begin{proof}[Proof of Proposition \ref{Prop_VAoInodei}]
		Node $i$ transmits updates to node $i+1$ in every time slot. Upon successful reception, node $i+1$ retains only the most recent version, discarding all earlier ones. Thus, the VAoI at node $i+1$ depends on the most recent successful transmissions from node $i$. If the latest transmission at time $t$ succeeds (w.p. $\rho_i$), node $i+1$ obtains the version held by node $i$ in the previous slot, i.e., $V_{i+1}(t) = V_i(t-1)$. The VAoI $\Delta_{i+1}(t) = V_S(t) - V_{i+1}(t)$ is then:
		\begin{align}
			\label{eqn_ProofVAoI_ip1_first}
			\Delta_{i+1}(t) = \underbrace{V_S(t) - V_S(t\!-\!1)}_{\eta_{1}} + \underbrace{V_S(t\!-\!1) - V_{i}(t-1)}_{\Delta_{i}(t-1)},
		\end{align}
		where $\eta_1 = V_S(t) - V_S(t-1)$ represents the number of new versions generated by the source in the most recent slot (either $0$ or $1$).
		If the latest transmission fails (w.p. $1 - \rho_i$), but the previous one at $t-1$ succeeds (w.p. $\rho_i$), then $V_{i+1}(t) = V_i(t-2)$, and the VAoI becomes:
		\begin{align}
			\label{eqn_ProofVAoI_ip1_second}
			\Delta_{i+1}(t) 
			&= \underbrace{V_S(t) - V_S(t\!-\!2)}_{\eta_{2}} + \underbrace{V_S(t\!-\!2) - V_{i}(t\!-\!2)}_{\Delta_{i}(t-2)}, 
		\end{align}
		where $\eta_2$ denotes the number of versions generated over the past two slots. Since the source generates a new version in each slot according to a Bernoulli process with parameter $p_g$, $\eta_k \sim \text{Bin}(k, p_g)$ over $k$ slots. Generally, the VAoI at node $i+1$ at time $t$ equals the VAoI at node $i$ from $m_i$ slots earlier plus the number of generated versions in those $m_i$ slots, where $m_i$ follows a Geometric distribution representing the number of transmissions required for successful delivery over link $i$.
	\end{proof}

	\section{}
	\label{Appen_Proof_VAoIDestNode}
    
	\begin{proof}[Proof of Lemma \ref{Lemma_VAoIDestNode}]
		According to Proposition \ref{Prop_VAoInodei}, the VAoI at all nodes can be expressed recursively as follows:
		\begin{align*}
        \scalebox{0.8}{$
			\begin{cases}
                \!\Delta_2(t)\!=\!\Delta_1(t\!-\!m_1)\!+\!\eta_{m_{1}}, \\
                \vdots \\
				\!\Delta_N(t)\!=\!\Delta_{N\!-\!1}(t\!-\!m_{N\!-\!1})\!+\!\eta_{m_{N\!-\!1}}, \\
				\!\Delta_{N\!+\!1}(t)\!=\!\Delta_{N}(t\!-\!m_{N})\!+\!\eta_{m_{N}},   
			\end{cases} 
            \hspace{-10pt} \! \Rightarrow \! \Delta_{N\!+\!1}(t)\!=\!\Delta_{1} (t\!-\!\!\underbrace{\sum_{i=1}^{N} m_{i}}_{=\tau_{N}})\!+\!\!\underbrace{\sum_{i=1}^{N} \eta_{m_i}}_{=\beta_{N}}\!.
            $}
		\end{align*}
		
		The expected value of $\tau_N$ and $\beta_N$ is derived:
		\begin{align*}
			\mathbb{E} \left[\tau_N \right] &\!=\!\! \sum_{i=1}^{N} \mathbb{E} \left[ m_{i} \right] \!=\!\! \sum_{i=1}^{N} \frac{1}{\rho_i}, \\
			\mathbb{E} \left[\beta_N \right] &\!=\! \!\sum_{i=1}^{N} \mathbb{E}  \left[ \eta_{m_i} \right] \!\overset{(a)}{=}\!  \sum_{i=1}^{N} \mathbb{E}_{m_i} \big[ \mathbb{E} \left[ \eta_{m_i} | m_i \right] \big] \!=\! p_g \! \sum_{i=1}^{N} \mathbb{E} \left[  m_i \right],
		\end{align*}
		where $(a)$ follows from the tower rule: $\mathbb{E} \left[ X \right] \!=\! \mathbb{E}_{Y} \!\big[ \mathbb{E}_{X|Y} \left[ X|Y \right] \big]$. Note that $\eta_{m_i} \!\mid\! m_i$ follows a Binomial distribution, as given in \eqref{eq_BinomialPMF}, with mean $m_i p_g$.
	\end{proof}

	\section{}
	\label{Appen_Proof_AvgVAoIlastNode}
	
	\begin{proof}[Proof of Theorem \ref{Lemma_AvgVAoIlastNode}]
		According to Lemma \ref{Lemma_VAoIDestNode},
		\begin{align}
			\bar{\Delta}_{N+1}(t) & \!=\! \mathbb{E} \left[ \Delta_{N+1}(t) \right] \!=\! \mathbb{E} \left[ \Delta_{1} (t\!-\!\tau_{N})\!+\!\beta_{N} \right]\\
			&\!\overset{(a)}{=}\! {\sum_{\tau = N}^{\infty} \!\mathbb{P}(\tau_N\!=\!\tau) \mathbb{E} \left[ \Delta_{1}(t\!-\!\tau) \right]} \!+\! p_g \!\sum_{i=1}^{N} \frac{1}{\rho_i},  \notag
		\end{align}
        
		\noindent where $(a)$ follows from the tower rule: 
		$\mathbb{E} \!\left[ \Delta_{1} \!(t\!-\!\tau_{N})\right] \!=\! \mathbb{E}_{\tau_N} \!\big[ \mathbb{E} \left[ \Delta_{1}\!(t\!-\!\tau_N) | \tau_N \right] \big]
			\!\!=\!\!\!\sum_{\tau = N}^{\infty} \!\mathbb{P}(\tau_N\!\!=\!\!\tau) \mathbb{E} \!\left[ \Delta_{1}\!(t\!-\!\tau) \right]\!.$
		The steady-state value $\bar{\Delta}_{N\!+\!1} \!=\! \lim_{t \rightarrow \infty} \bar{\Delta}_{N+1}(t)$ is then:
        
		\begin{align*}
			\bar{\Delta}_{N\!+\!1} \!\!\overset{(b)}{=}\!\! \lim_{t \rightarrow \infty} \!\!\mathbb{E} \!\left[ \Delta_{1}\!(t) \right] \!\!\sum_{\tau = N}^{\infty} \!\!\mathbb{P}(\tau_N\!=\!\tau)  \!+\! p_g \!\!\sum_{i=1}^{N} \!\!\frac{1}{\rho_i}
			 \!=\! \bar{\Delta}_{1} \!\!+\! p_g \!\!\sum_{i=1}^{N} \!\!\frac{1}{\rho_i}\!.
		\end{align*}
		
		Equality $(b)$ follows directly from the relation $\lim_{t \rightarrow \infty} \mathbb{E}[\Delta_1(t-\tau)] = \lim_{t \rightarrow \infty} \mathbb{E}[\Delta_1(t)]$, which holds for an ergodic and integrable \textsc{dtmc} $\Delta_1(t)$, where $\mathbb{E}\big[|\Delta_1(t)|\big] < \infty$~\cite[Sec. 1.10]{norris1998markov}. This condition is satisfied by the \textsc{dtmc}s presented in Section \ref{Sec_OneHop} under all three policies, for which the steady-state distribution is stationary and therefore time-invariant.
	\end{proof}
	
	\balance
	\bibliographystyle{IEEEtran}
	\bibliography{Refs}

@article{kountouris2021semantics,
  title={Semantics-empowered communication for networked intelligent systems},
  author={Kountouris, Marios and Pappas, Nikolaos},
  journal={IEEE Communications Magazine},
  volume={59},
  number={6},
  pages={},
  year={2021},
  publisher={}
}

@article{yates2021age,
  title={Age of information: An introduction and survey},
  author={Yates, Roy D and Sun, Yin and Brown, D Richard and Kaul, Sanjit K and Modiano, Eytan and Ulukus, Sennur},
  journal={IEEE Journal on Selected Areas in Communications},
  volume={39},
  number={5},
  pages={},
  year={2021},
  publisher={}
}

@book{altman1999constrained,
  title={Constrained Markov Decision Processes},
  author={Altman, Eitan},
  volume={7},
  year={1999},
  publisher={CRC Press}
}

@article{beutler1985optimal,
  title={Optimal policies for controlled Markov chains with a constraint},
  author={Beutler, Frederick J and Ross, Keith W},
  journal={Journal of mathematical analysis and applications},
  volume={112},
  number={1},
  pages={},
  year={1985},
  publisher={}
}

@book{norris1998markov,
  title={Markov chains},
  author={Norris, James R},
  number={2},
  year={1998},
  publisher={Cambridge university press}
}

@book{bertsekas2011dynamic,
  author = {Bertsekas, Dimitri P.},
  title = {Dynamic Programming and Optimal Control, Vol. II},
  year = {2007},
  publisher = {Athena Scientific},
  edition = {3rd},
}

@inproceedings{champati2019distribution,
  title={On the distribution of {AoI} for the {GI/GI/1/1} and {GI/GI/1/2} systems: Exact expressions and bounds},
  author={Champati, Jaya Prakash and Al-Zubaidy, Hussein and Gross, James},
  booktitle={IEEE Conference on Computer Communications (INFOCOM)},
  year={2019},
  volume={},
  number={},
  pages={37-45},
  organization={}
}

@article{ji2024age,
  title={Age-optimal packet scheduling with resource constraint and feedback delay},
  author={Ji, Yonghao and Lu, Yuxiao and Xu, Xiaoli and Huang, Xinmei},
  journal={IEEE Transactions on Communications},
  volume={72},
  number={7},
  pages={},
  year={2024},
  publisher={}
}

@article{ayan2020probability,
  title={Probability analysis of age of information in multi-hop networks},
  author={Ayan, Onur and G{\"u}rsu, H Murat and Papa, Arled and Kellerer, Wolfgang},
  journal={IEEE Networking Letters},
  volume={2},
  number={2},
  pages={},
  year={2020},
  publisher={}
}

@article{inoue2019general,
  title={A general formula for the stationary distribution of the age of information and its application to single-server queues},
  author={Inoue, Yoshiaki and Masuyama, Hiroyuki and Takine, Tetsuya and Tanaka, Toshiyuki},
  journal={IEEE Transactions on Information Theory},
  volume={65},
  number={12},
  pages={},
  year={2019},
  publisher={}
}

@article{wang2021age,
  title={Age of changed information: Content-aware status updating in the Internet of Things},
  author={Wang, Xijun and Lin, Wenrui and Xu, Chao and Sun, Xinghua and Chen, Xiang},
  journal={IEEE Transactions on Communications},
  volume={70},
  number={1},
  pages={},
  year={2021},
  publisher={}
}

@article{jiang2021joint,
  title={Joint performance analysis of ages of information in a multi-source pushout server},
  author={Jiang, Yukang and Miyoshi, Naoto},
  journal={IEEE Transactions on Information Theory},
  volume={68},
  number={2},
  pages={965--975},
  year={2021},
  publisher={}
}

@article{akar2025age,
  title={Age of information in a single-source generate-at-will dual-server status update system},
  author={Akar, Nail and Ulukus, Sennur},
  journal={IEEE Transactions on Communications},
  year={2025},
  volume={73},
  number={9},
  pages={7431-7444},
  publisher={}
}

@article{salimnejad2024age,
  title={Age of information versions: A semantic view of markov source monitoring},
  author={Salimnejad, Mehrdad and Kountouris, Marios and Ephremides, Anthony and Pappas, Nikolaos},
  journal={IEEE Transactions on Communications}, 
  year={2025},
  volume={73},
  number={12},
  pages={14486-14502}
}

@article{costa2016age,
  title={On the age of information in status update systems with packet management},
  author={Costa, Maice and Codreanu, Marian and Ephremides, Anthony},
  journal={IEEE Transactions on Information Theory},
  volume={62},
  number={4},
  pages={},
  year={2016},
  publisher={}
}

@article{maatouk2020age,
  title={The age of incorrect information: A new performance metric for status updates},
  author={Maatouk, Ali and Kriouile, Saad and Assaad, Mohamad and Ephremides, Anthony},
  journal={IEEE/ACM Transactions on Networking},
  volume={28},
  number={5},
  pages={},
  year={2020},
  publisher={}
}

@inproceedings{kaul2012real,
  title={Real-time status: How often should one update?},
  author={Kaul, Sanjit and Yates, Roy and Gruteser, Marco},
  booktitle={IEEE Conference on Computer Communications (INFOCOM)},
  volume={},
  number={},
  pages={2731-2735},
  year={2012},
  organization={}
}

@inproceedings{yates2021vage,
  title={The age of gossip in networks},
  author={Yates, Roy D},
  booktitle={IEEE International Symposium on Information Theory (ISIT)},
  year={2021},
  volume={},
  number={},
  pages={2984-2989},
  organization={}
}

@article{chen2024minimizing,
  title={Minimizing age of incorrect information over a channel with random delay},
  author={Chen, Yutao and Ephremides, Anthony},
  journal={IEEE/ACM Transactions on Networking},
  volume={32},
  number={4},
  pages={},
  year={2024},
  publisher={}
}

@article{karevvanavar2024version,
  title={Version age of information minimization over fading broadcast channels},
  author={Karevvanavar, Gangadhar and Pable, Hrishikesh and Patil, Om and Bhat, Rajshekhar V and Pappas, Nikolaos},
  journal={IEEE Transactions on Wireless Communications},
  year={2025},
  volume={24},
  number={2},
  pages={1620-1634},
}

@article{abd2022closed,
  title={Closed-form characterization of the {MGF} of {AoI} in energy harvesting status update systems},
  author={Abd-Elmagid, Mohamed A and Dhillon, Harpreet S},
  journal={IEEE Transactions on Information Theory},
  volume={68},
  number={6},
  pages={},
  year={2022},
  publisher={}
}

@article{inoue2025characterizing,
  title={Characterizing the Age of Information with Multiple Coexisting Data Streams},
  author={Inoue, Yoshiaki and Mandjes, Michel},
  journal={IEEE Transactions on Information Theory},
  year={2025},
  volume={71},
  number={6},
  pages={4732-4753},
  publisher={}
}

@article{moltafet2022moment,
  title={Moment generating function of age of information in multisource {M/G/1/1} queueing systems},
  author={Moltafet, Mohammad and Leinonen, Markus and Codreanu, Marian},
  journal={IEEE Transactions on Communications},
  volume={70},
  number={10},
  pages={},
  year={2022},
  publisher={}
}

@article{fiems2023age,
  title={Age of information analysis with preemptive packet management},
  author={Fiems, Dieter},
  journal={IEEE Communications Letters},
  volume={27},
  number={4},
  pages={},
  year={2023},
  publisher={}
}

@article{akar2021discrete,
  title={Discrete-time queueing model of age of information with multiple information sources},
  author={Akar, Nail and Dogan, Ozancan},
  journal={IEEE Internet of Things Journal},
  volume={8},
  number={19},
  pages={},
  year={2021},
  publisher={}
}

@article{kosta2021age,
  title={The age of information in a discrete time queue: Stationary distribution and non-linear age mean analysis},
  author={Kosta, Antzela and Pappas, Nikolaos and Ephremides, Anthony and Angelakis, Vangelis},
  journal={IEEE Journal on Selected Areas in Communications},
  volume={39},
  number={5},
  pages={1352--1364},
  year={2021},
  publisher={}
}

@inproceedings{zhang2021age,
  title={On age of information for discrete time status updating system with {Ber/G/1/1} queues},
  author={Zhang, Jixiang and Xu, Yinfei},
  booktitle={IEEE Information Theory Workshop (ITW)},
  pages={},
  year={2021},
  organization={}
}

@article{yates2020age,
  title={The age of information in networks: Moments, distributions, and sampling},
  author={Yates, Roy D},
  journal={IEEE Transactions on Information Theory},
  volume={66},
  number={9},
  pages={},
  year={2020},
  publisher={}
}

@inproceedings{delfani2024semantics,
  title={Semantics-aware status updates with energy harvesting devices: Query version age of information},
  author={Delfani, Erfan and Pappas, Nikolaos},
  booktitle={22nd WiOpt Conference},
  year={2024},
  volume={},
  number={},
  pages={177-184},
}

@article{Chiariotti2021PAoI,
  author={Chiariotti, Federico and Vikhrova, Olga and Soret, Beatriz and Popovski, Petar},
  journal={IEEE Transactions on Communications}, 
  title={Peak Age of Information Distribution for Edge Computing With Wireless Links}, 
  year={2021},
  volume={69},
  number={5},
  pages={},
}

@article{Zhang2025AoIVehicles,
  author={Zhang, Tianci and Chen, Zhengchuan and Tian, Zhong and Wang, Min and Zhen, Li and Wu, Dapeng Oliver and Li, Yonghui and Quek, Tony Q. S.},
  journal={IEEE Transactions on Communications}, 
  title={Age of Information in Internet of Vehicles: A Discrete-Time Multisource Queueing Model}, 
  year={2025},
  volume={73},
  number={5},
  pages={},
}

@article{Tripathi2023Mhop,
  author={Tripathi, Vishrant and Talak, Rajat and Modiano, Eytan},
  journal={IEEE/ACM Transactions on Networking}, 
  title={Information Freshness in Multihop Wireless Networks}, 
  year={2023},
  volume={31},
  number={2},
  pages={},
}

@inproceedings{Vikhrova2020Mhop,
  author={Vikhrova, Olga and Chiariotti, Federico and Soret, Beatriz and Araniti, Giuseppe and Molinaro, Antonella and Popovski, Petar},
  booktitle={IEEE Global Communications Conference (GLOBECOM)}, 
  title={Age of Information in Multi-hop Networks with Priorities}, 
  year={2020},
  volume={},
  number={},
  pages={},
}

@inproceedings{Talak2018Mhop,
  author={Talak, Rajat and Karaman, Sertac and Modiano, Eytan},
  booktitle={55th Annual Allerton Conference on Communication, Control, and Computing (Allerton)}, 
  title={Minimizing age-of-information in multi-hop wireless networks}, 
  year={2017},
  volume={},
  number={},
  pages={486-493},
}

@inproceedings{Sinha2024Tandem,
  author={Sinha, Ashirwad and Singhvi, Shubhransh and Mankar, Praful D. and Dhillon, Harpreet S.},
  booktitle={IEEE International Symposium on Information Theory (ISIT)}, 
  title={Peak Age of Information under Tandem of Queues}, 
  year={2024},
  volume={},
  number={},
  pages={951-956},
}

@article{Bedewy2019Mhop,
  author={Bedewy, Ahmed M. and Sun, Yin and Shroff, Ness B.},
  journal={IEEE/ACM Transactions on Networking}, 
  title={The Age of Information in Multihop Networks}, 
  year={2019},
  volume={27},
  number={3},
  pages={},
}

@article{Chiariotti2022Mhop,
  author={Chiariotti, Federico and Vikhrova, Olga and Soret, Beatriz and Popovski, Petar},
  journal={IEEE Transactions on Communications}, 
  title={Age of Information in Multihop Connections With Tributary Traffic and No Preemption}, 
  year={2022},
  volume={70},
  number={10},
  pages={},
}

@inproceedings{Kaswan2023Mhop,
  author={Kaswan, Priyanka and Ulukus, Sennur},
  booktitle={IEEE International Symposium on Information Theory (ISIT)}, 
  title={Age of Information With Non-Poisson Updates in Cache-Updating Networks}, 
  year={2023},
  volume={},
  number={},
  pages={957-962},
}

@inproceedings{asvadi2024age,
  title={Age of Information in Multipath Multihop Networks},
  author={Asvadi, Sepehr and Ashtiani, Farid},
  booktitle={IEEE Wireless Communications and Networking Conference (WCNC)},
  pages={},
  year={2024},
  organization={}
}

@article{buyukates2019age,
  title={Age of information in multihop multicast networks},
  author={Buyukates, Baturalp and Soysal, Alkan and Ulukus, Sennur},
  journal={Journal of Communications and Networks},
  volume={21},
  number={3},
  pages={256--267},
  year={2019},
  publisher={}
}

@article{Delfani2025LEO,
  author={Delfani, Erfan and Pappas, Nikolaos},
  journal={IEEE Communications Letters}, 
  title={Semantics-Aware Updates from Remote Energy Harvesting Devices to Interconnected {LEO} Satellites}, 
  year={2025},
  volume={29},
  number={8},
  pages={1928-1932},
}

@ARTICLE{Mehrdad2025CL,
  author={Salimnejad, Mehrdad and Pappas, Nikolaos and Kountouris, Marios},
  journal={IEEE Communications Letters}, 
  title={So Timely, Yet So Stale: The Impact of Clock Drift in Real-Time Systems}, 
  year={2025},
  volume={29},
  number={10},
  pages={2228-2232},
  doi={10.1109/LCOMM.2025.3590865}}

@article{luo2025survey,
  title={From Information Freshness to Semantics of Information and Goal-oriented Communications},
  author={Luo, Jiping and Delfani, Erfan and Salimnejad, Mehrdad and Pappas, Nikolaos},
  journal={arXiv:2512.12758},
  year={2025}
}
	
\end{document}